\documentclass{llncs}

\usepackage{amsmath}
\usepackage{amssymb}
\usepackage{clrscode4e}
\usepackage{algpseudocode}
\usepackage{algorithm}
\usepackage{graphicx}
\usepackage{longtable}

\newcommand\remove[1]{}
\spnewtheorem{observation}{Observation}{\bfseries}{\itshape}

\title{Improved Approximation Algorithms for Projection Games}
\author{Pasin Manurangsi\inst{1}\thanks{Part of this work was completed while the author was at Massachusetts Institute of Technology.} \and Dana Moshkovitz\inst{2} \thanks{This material is based upon work supported by the National Science Foundation under Grant numbers 1218547 and 1452302.}}
\date{\today}

\institute{University of California, Berkeley, CA 94720, USA \\
        \email{pasin@berkeley.edu}
    	\and
        Massachusetts Institute of Technology, Cambridge MA 02139, USA\\
  \email{dmoshkov@mit.edu}}

\begin{document}
\maketitle

\begin{abstract}
The projection games (aka \textsc{Label Cover}) problem is of great importance to the field of approximation algorithms, since most of the NP-hardness of approximation results we know today are reductions from \textsc{Label Cover}.
In this paper we design several approximation algorithms for projection games:
\begin{enumerate}
\item	A polynomial-time approximation algorithm that improves on the previous best approximation by Charikar, Hajiaghayi and Karloff~\cite{CHK}.

\item	A sub-exponential time algorithm with much tighter approximation for the case of smooth projection games.

\item A polynomial-time approximation scheme (PTAS) for projection games on planar graphs and a tight running time lower bound for such approximation schemes.\footnote{The conference version of this paper had only the PTAS but not the running time lower bound.} \\
\end{enumerate}

\end{abstract}

\section{Introduction}

The projection games problem (also known as \textsc{Label Cover}) can be defined as follows.

\textsc{Input:}
A bipartite graph $G = (A, B, E)$, two finite sets of alphabets (aka labels) $\Sigma_A,\Sigma_B$, and, for each edge $e = (a, b) \in E$, a ``projection'' $\pi_e : \Sigma_A \to \Sigma_B$.

\textsc{Goal:}
Find an assignment to the vertices $\varphi_A : A \to \Sigma_A$ and $\varphi_B : B \to \Sigma_B$ that maximizes the number of edges $e = (a, b)$ that are ``satisfied'', i.e., $\pi_e(\varphi_A(a)) = \varphi_B(b)$.\\

An instance is said to be ``satisfiable'' or ``feasible'' or have ``perfect completeness'' if there exists an assignment that satisfies all edges. An instance is said to be ``$\delta$-nearly satisfiable'' or ``$\delta$-nearly feasible'' if there exists an assignment that satisfies $(1-\delta)$ fraction of the edges. In this work, we focus on satisfiable instances of projection games.


\textsc{Label Cover} has gained much significance for approximation algorithms because of the following PCP Theorem, establishing that it is NP-hard, given a satisfiable projection game instance, to satisfy even an $\varepsilon$ fraction of the edges:
\begin{theorem}[Strong PCP Theorem]\label{t:strong-PCP}
For every $n$, $\varepsilon = \varepsilon(n)$, there is $k = k(\varepsilon)$, such that deciding \textsc{Sat} on inputs of size $n$ can be reduced to finding, given a satisfiable projection game on alphabets of size $k$, an assignment that satisfies more than an $\varepsilon$ fraction of the edges.
\end{theorem}
This theorem is the starting point of the extremely successful long-code based framework for achieving hardness of approximation results~\cite{BGS,Has97}, as well as of other optimal hardness of approximation results, e.g., for \textsc{Set Cover}~\cite{Feige,M,DS}.

We know several proofs of the strong PCP theorem that yield different parameters in Theorem~\ref{t:strong-PCP}. The parallel repetition theorem~\cite{Raz}, applied on the basic PCP Theorem~\cite{BFL,BFLS,AS,ALMSS}, yields  $k(\varepsilon) = (1/\varepsilon)^{O(1)}$. Alas, it reduces exact \textsc{Sat} on input size $n$ to \textsc{Label Cover} on input size $n^{O(\log 1/\varepsilon)}$. Hence, a lower bound of $2^{\Omega(n)}$ for the time required for solving \textsc{Sat} on inputs of size $n$ only implies a lower bound of $2^{n^{\Omega(1/\log 1/\varepsilon)}}$ for \textsc{Label Cover} via this theorem. This bound deteriorates as $\varepsilon$ approaches zero; for instance, if $\varepsilon = (1/n)^{O(1)}$, then the bound is $\Omega(1)$, which gives us no information at all.

A different proof is based on PCP composition~\cite{MR08,DS}. It has smaller blow up but larger alphabet size. Specifically, it shows a reduction from exact \textsc{Sat} with input size $n$ to \textsc{Label Cover} with input size $n^{1+o(1)}poly(1/\varepsilon)$ and alphabet size $\exp(1/\varepsilon)$.

One is tempted to conjecture that a PCP theorem with both a blow-up of $n^{1+o(1)}poly(1/\varepsilon)$ and an alphabet size $(1/\varepsilon)^{O(1)}$ holds. See~\cite{M} for a discussion of potential applications of this ``Projection Games Conjecture''.

Finding algorithms for projection games is therefore both a natural pursuit in combinatorial optimization, and also a way to advance our understanding of the main paradigm for settling the approximability of optimization problems.
Specifically, our algorithms help make progress towards the following questions:
\begin{enumerate}
\item Is the "Projection Games Conjecture" true? What is the tradeoff between the alphabet size, the blow-up and the approximation factor?
\item What about even stronger versions of the strong PCP theorem?
E.g., Khot introduced ``smooth'' projection games~\cite{Khot-coloring} (see discussion below for the definition). What kind of lower bounds can we expect to get via such a theorem?
\item	Does a strong PCP theorem hold for graphs of special forms, e.g., on planar graphs?
\end{enumerate}

\section{Our Results}

\subsection{Better Approximation in Polynomial Time}

In 2009, Charikar, Hajiaghayi and Karloff presented a polynomial-time $O((nk)^{1/3})$-approximation algorithm for \textsc{Label Cover} on graphs with $n$ vertices and alphabets of size $k$~\cite{CHK}.\footnote{Recall that, for $\alpha \geq 1$, an $\alpha$-approximation algorithm for a maximization problem is an algorithm that, for every input, outputs a solution of value at least $1/\alpha$ times the value of the optimal solution. $\alpha$ is called the ``approximation ratio'' of the algorithm.} This improved on Peleg's $O((nk)^{1/2})$-approximation algorithm~\cite{Pel07}. Both Peleg's and the CHK algorithms worked in the more general setting of arbitrary (not necessarily projections) constraints on the edges and possibly unsatisfiable instances.
We show a polynomial-time algorithm that achieves a better approximation for satisfiable projection games:
\begin{theorem}\label{t:approx}
There is a polynomial-time algorithm that, given a satisfiable instance of projection games on a graph of size $n$ and alphabets of size $k$, finds an assignment that satisfies $\Omega(1/(nk)^{1/4})$ fraction of the edges.
\end{theorem}

\subsection{Algorithms for Smooth Projection Games}

Khot introduced ``smooth'' projection games in order to obtain new hardness of approximation results, e.g., for coloring problems~\cite{Khot-coloring}. In a smooth projection game, for every vertex $a\in A$, the assignments projected to $a$'s neighborhood by the different possible assignments $\sigma_a\in\Sigma_A$ to $a$, differ a lot from one another (alternatively, form an error correcting code with high relative distance). More formally:
  \begin{definition} \label{d:smooth}
    A projection game instance is $\mu$-smooth if for every $a \in A$ and any distinct assignments $\sigma_a, \sigma'_a \in\Sigma_A$, we have $$Pr_{b \in N(a)}[\pi_{(a, b)}(\sigma_a) = \pi_{(a, b)}(\sigma'_a)] \leq \mu$$
  \end{definition}
where $N(a)$ is the set of neighbors of $a$.

Intuitively, smoothness makes the projection games problem easier, since knowing only a small fraction of the assignment to a neighborhood of a vertex $a\in A$  determines the assignment to $a$.

Smoothness can be seen as an intermediate property between projection and uniqueness, with uniqueness being $0$-smoothness. Khot's Unique Games Conjecture~\cite{Khot} is that the Strong PCP Theorem holds for unique games on nearly satisfiable instances for any constant $\varepsilon >0$.

The Strong PCP Theorem (Theorem~\ref{t:strong-PCP}) is known to hold for $\mu$-smooth projection games with $\mu > 0$. However, the known reductions transform \textsc{Sat} instances of size $n$ to instances of smooth \textsc{Label Cover} of size at least $n^{O((1/\mu)\log (1/\varepsilon))}$~\cite{Khot-coloring,HK}. Hence, a lower bound of $2^{\Omega(n)}$ for \textsc{Sat} only translates into a lower bound of $2^{n^{\Omega(\mu/\log (1/\varepsilon))}}$ for $\mu$-smooth projection games.

Interestingly, the efficient reduction of Moshkovitz and Raz~\cite{MR08} inherently generates instances that are not smooth. Moreover, for unique games it is known that if they admit a reduction from \textsc{Sat} of size $n$, then the reduction must incur a blow-up of at least $n^{1/\delta^{\Omega(1)}}$ for $\delta$-almost satisfiable instances. This follows from the sub-exponential time algorithm of Arora, Barak and Steurer~\cite{ABS}.

Given this state of affairs, one wonders whether a large blow-up is necessary for smooth projection games.
We make progress toward settling this question by showing:
\begin{theorem}\label{t:smooth}
For any constant $c\geq 1$, the following holds: there is a randomized algorithm that given a $\mu$-smooth satisfiable projection game in which all vertices in $A$ have degrees at least $\frac{c\log |A|}{\mu}$, finds an optimal assignment in time $\exp(O(\mu |B|\log |\Sigma_B|))poly(|A|,|\Sigma_A|)$ with probability $1/2$.

Moreover, there is a deterministic $O(1)$-approximation algorithm for $\mu$-smooth satisfiable projection games of any degree. The deterministic algorithm runs in time $\exp(O(\mu |B|\log |\Sigma_B|))poly(|A|,|\Sigma_A|)$ as well.
\end{theorem}
The algorithms work by finding a subset of fraction $\mu$ in $B$ that is connected to all, or most, of the vertices in $A$ and going over all possible assignments to it.

Theorem~\ref{t:smooth} essentially implies that a blow-up of $n/\mu$ is necessary for any reduction from \textsc{Sat} to $\mu$-smooth \textsc{Label Cover}, no matter what is the approximation factor $\varepsilon$.

\subsection{PTAS For Planar Graphs}

As the strong PCP Theorem (Theorem~\ref{t:strong-PCP}) shows, \textsc{Label Cover} is NP-hard to approximate even to within subconstant $\varepsilon$ fraction. Does \textsc{Label Cover} remain as hard when we consider special kinds of graphs?

In recent years there has been much interest in optimization problems over {\em planar graphs}.
These are graphs that can be embedded in the plane without edges crossing each other. Many optimization problems have very efficient algorithms on planar graphs.

We show that while projection games remain NP-hard to solve {\em exactly} on planar graphs, when it comes to approximation, they admit a PTAS:
\begin{theorem}\label{t:planar} The following hold:
\begin{enumerate}
\item
Given a satisfiable instance of projection games on a planar graph, it is NP-hard to find a satisfying assignment.
\item
There is a polynomial time approximation scheme for projection games on planar graphs that runs in time $(nk)^{O(1/\varepsilon)}$. Moreover, this running time is essentially tight: there is no PTAS for projection games on planar graphs running in time $2^{O(1/\varepsilon)^\gamma}(nk)^{O(1/\varepsilon)^{1-\delta}}$ for any constants $\gamma, \delta > 0$ unless the exponential time hypothesis (ETH) \footnote{Note that the exponential time hypothesis states that {\sc 3-SAT} cannot be solved in sub-exponential time.} fails.
\end{enumerate}
\end{theorem}

The NP-hardness of projection games on planar graphs is based on a reduction from 3-colorability problem on planar graphs. The PTAS works via Baker's approach~\cite{Baker94} of approximating the graph by a graph with constant tree-width. Finally, the running time lower bound is shown via a reduction from from {\sc Matrix Tiling} problem introduced by Marx in \cite{Mar07}.

\section{Conventions}

We define the following notation to be used in the paper.
\begin{itemize}
\item Let $n_A = |A|$ denote the number of vertices in $A$ and $n_B = |B|$ denote the number of vertices in $B$. Let $n$ denote the number of vertices in the whole graph, i.e. $n = n_A + n_B$.
\item Let $d_v$ denote the degree of a vertex $v \in A \cup B$.
\item For a vertex $u$, we use $N(u)$ to denote set of vertices that are neighbors of $u$. Similarly, for a set of vertex $U$, we use $N(U)$ to denote the set of vertices that are neighbors of at least one vertex in $U$.
\item For each vertex $u$, define $N_2(u)$ to be $N(N(u))$. This is the set of neighbors of neighbors of $u$.
\item Let $\sigma_v^{OPT}$ be the assignment to $v$ in an assignment to vertices that satisfies all the edges. In short, we will sometimes refer to this as ``the optimal assignment''. This is guaranteed to exist from our assumption that the instances considered are satisfiable.
\item For any edge $e = (a, b)$, we define $p_e$ to be $|\pi^{-1}(\sigma_b^{OPT})|$. In other words, $p_e$ is the number of assignments to $a$ that satisfy the edge $e$ given that $b$ is assigned $\sigma_b^{OPT}$, the optimal assignment. Define $\overline{p}$ to be the average of $p_e$ over all $e$; that is $\overline{p} = \frac{\sum_{e \in E} p_e}{|E|}$.
\item For each set of vertices $S$, define $E(S)$ to be the set of edges of $G$ with at least one endpoint in $S$, i.e., $E(S) = \{(u, v) \in E \mid u \in S \text{ or } v \in S\}$.
\item For each $a \in A$, let $h(a)$ denote $|E({N_2(a)})|$. Let $h_{max} = max_{a\in A} h(a)$.
\end{itemize}
For simplicity, we make the following assumptions:
\begin{itemize}
\item $G$ is connected. This assumption can be made without loss of generality, as, if $G$ is not connected, we can always perform any algorithm presented below on each of its connected components and get an equally good or a better approximation ratio.
\item For every $e\in E$ and every $\sigma_b\in\Sigma_B$, the number of preimages in $\pi^{-1}_e(\sigma_b)$ is the same. In particular, $p_e = \overline{p}$ for all $e \in E$.
\end{itemize}

We only make use of the assumptions in the algorithms for proving Theorem~\ref{t:approx}.
We defer the treatment of graphs with general number of preimages to the appendix.

\section{Polynomial-time Approximation Algorithms for Projection Games}

In this section, we present an improved polynomial time approximation algorithm for projection games and prove Theorem~\ref{t:approx}.

To prove the theorem, we proceed to describe four polynomial-time approximation algorithms. In the end, by using the best of these four, we are able to produce a polynomial-time $O\left((n_A|\Sigma_A|)^{1/4}\right)$-approximation algorithm as desired. Next, we will list the algorithms along with its rough descriptions (see also illustrations in Figure~\ref{fig:algo} below); detailed description and analysis of each algorithm will follow later in this section:
\begin{enumerate}
\item {\bf Satisfy one neighbor -- $|E|/n_B$-approximation.} Assign each vertex in $A$ an arbitrary assignment. Each vertex in $B$ is then assigned to satisfy one of its neighboring edges. This algorithm satisfies at least $n_B$ edges.

\item {\bf Greedy assignment -- ${|\Sigma_A|}/{\overline{p}}$-approximation.} Each vertex in $B$ is assigned an assignment $\sigma_b\in\Sigma_B$ that has the largest number of preimages across neighboring edges $\sum_{a \in N(b)} |\pi_{(a, b)}^{-1}(\sigma_b)|$. Each vertex in $A$ is then assigned so that it satisfies as many edges as possible. This algorithm works well when $\Sigma_B$ assignments have many preimages.

\item {\bf Know your neighbors' neighbors -- $|E|\overline{p}/h_{max}$-approximation.} For a vertex $a_0 \in A$, we go over all possible assignments to it. For each assignment, we assign its neighbors $N(a_0)$ accordingly. Then, for each node in $N_2(a_0)$, we keep only the assignments that satisfy all the edges between it and vertices in $N(a_0)$.

When $a_0$ is assigned the optimal assignment, the number of choices for each node in $N_2(a_0)$ is reduced to at most $\overline{p}$ possibilities. In this way, we can satisfy $1/{\overline{p}}$ fraction of the edges that touch $N_2(a_0)$. This satisfies many edges when there exists $a_0\in A$ such that $N_2(a_0)$ spans many edges.

\item {\bf Divide and Conquer -- $O(n_A n_B h_{max}/|E|^2)$-approximation.} For every $a\in A$ we can fully satisfy $N(a) \cup N_2(a)$ efficiently, and give up on satisfying other edges that touch $N_2(a)$. Repeating this process, we can satisfy $\Omega(|E|^2/(n_A n_B h_{max}))$ fraction of the edges. This is large when $N_2(a)$ does not span many edges for all $a\in A$.
\end{enumerate}
The smallest of the four approximation factors is at most as large as their geometric mean, i.e., $$O\left(\sqrt[4]{\frac{|E|}{n_B}\cdot\frac{|\Sigma_A|}{\overline{p}}\cdot \frac{|E|\overline{p}}{h_{max}}\cdot\frac{n_A n_B h_{max}}{|E|^2}}\right)= O((n_A|\Sigma_A|)^{1/4}).$$

\begin{figure}[h!]
  \begin{center}
    \begin{tabular}{ c | c }
      \includegraphics[scale=0.20]{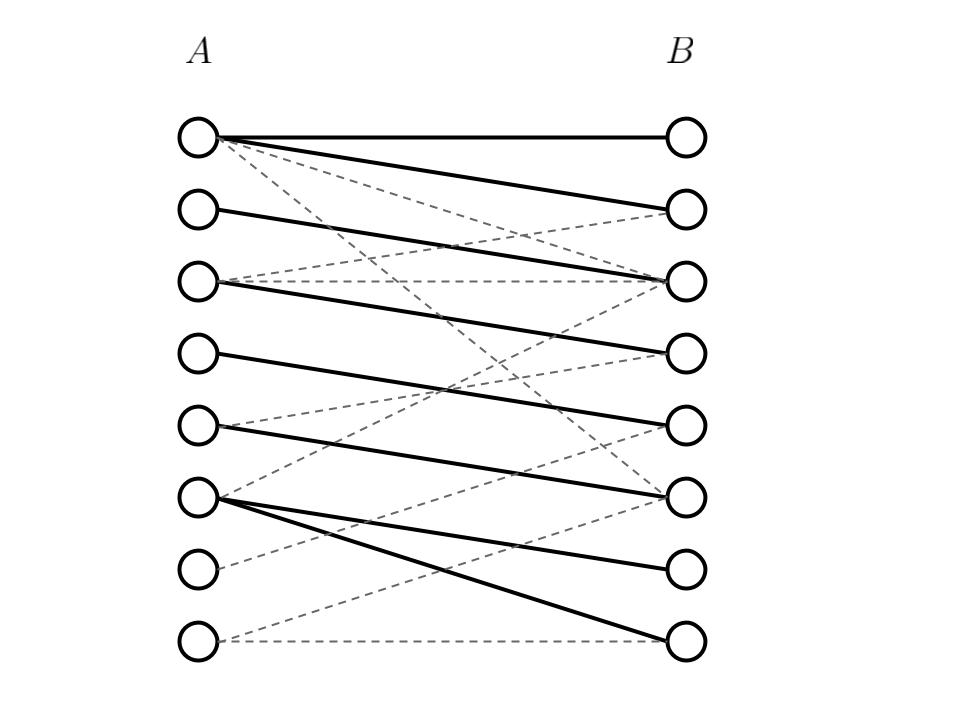} &  \includegraphics[scale=0.20]{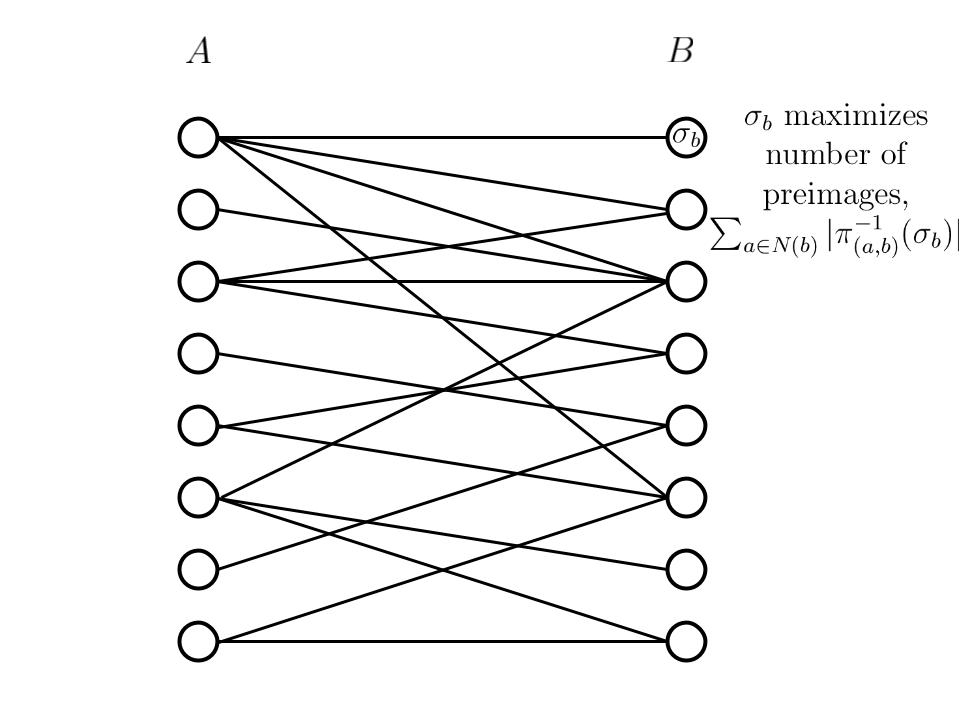} \\
      (1) & (2) \\
      \\
      \hline
      \\
      \includegraphics[scale=0.20]{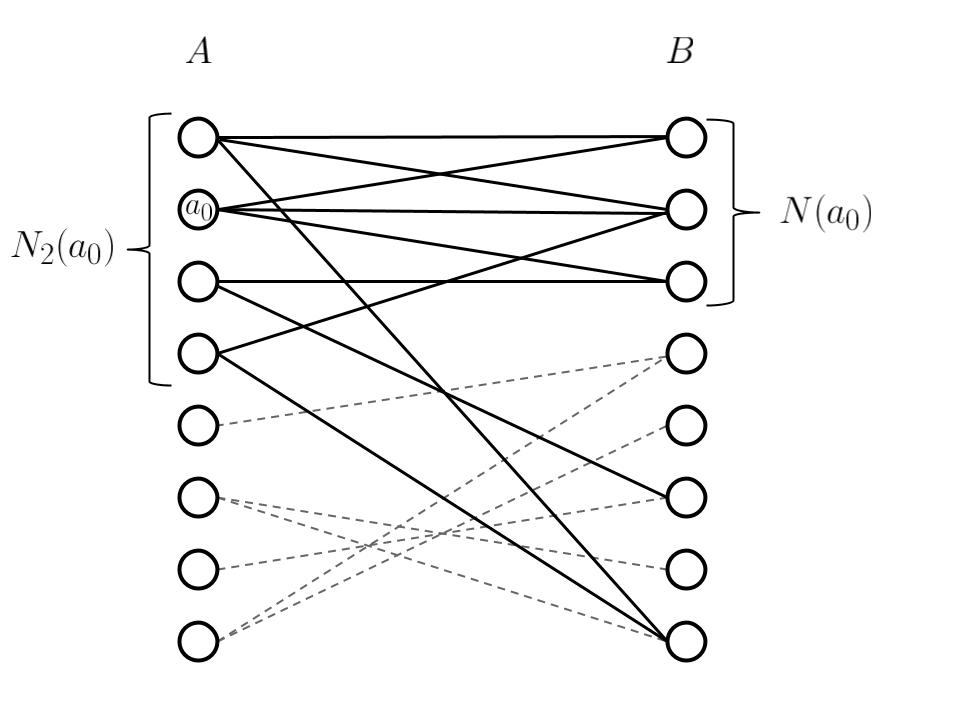} &  \includegraphics[scale=0.20]{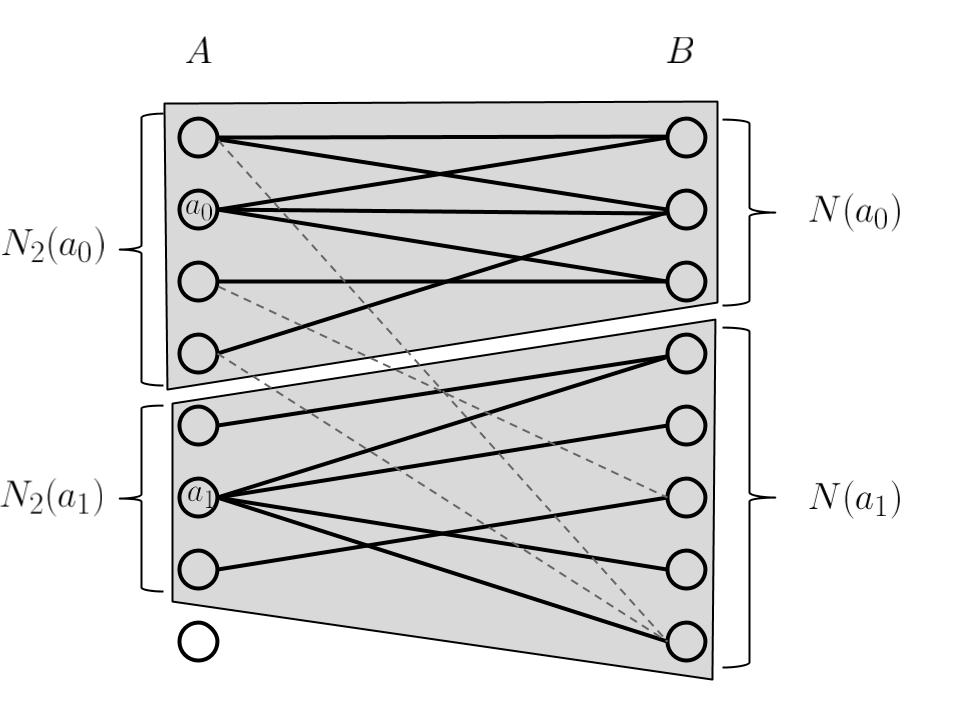} \\
      (3) & (4)
    \end{tabular}
  \end{center}
  \caption{
    \textbf{An Overview of The Algorithms in One Figure.} Four algorithms are used to prove Theorem~\ref{t:approx}: (1) In \emph{satisfy one neighbor} algorithm, each vertex in $B$ is assigned to satisfy one of its neighboring edges. (2) In \emph{greedy assignment} algorithm, each vertex in $B$ is assigned with an assignment with largest number of preimages. (3) In \emph{know your neighbors' neighbors} algorithm, for a vertex $a_0$, choices for each node in $N_2(a_0)$ are reduced to at most $O(\overline{p})$ possibilities so $O\left(\frac{1}{\overline{p}}\right)$ fraction of edges that touch $N_2(a_0)$ are satisfied. (4) In \emph{divide and conquer} algorithm, the vertices are seperated to subsets, each of which is a subset of $N(a) \cup N_2(a)$, and each subset is solved separately.
  }
  \label{fig:algo}
\end{figure}

All the details of each algorithm are described below.

\subsubsection{Satisfy One Neighbor Algorithm.}
We will present a simple algorithm that gives $\frac{|E|}{n_B}$ approximation ratio.

\begin{lemma} \label{dBapprox}
  For satisfiable instances of projection games, an assignment that satisfies at least $n_B$ edges can be found in polynomial time, which gives the approximation ratio of $\frac{|E|}{n_B}$.
\end{lemma}

\begin{proof}
  For each node $a \in A$, pick one $\sigma_a \in \Sigma_A$ and assign it to $a$. Then, for each $b \in B$, pick one neighbor $a$ of $b$ and assign $\varphi(b) = \pi_e(\sigma_a)$ for $b$.
  This guarantees that at least $n_B$ edges are satisfied.
\end{proof}

\subsubsection{Greedy Assignment Algorithm.}
The idea for this algorithm is that if there are many assignments in $\Sigma_A$ that satisfy each edge, then one satisfies many edges by guessing assignments at random. The algorithm below is the deterministic version of this algorithm.

\begin{lemma} \label{pickbest}
  There exists a polynomial-time $\frac{|\Sigma_A|}{\overline{p}}$-approximation algorithm for satisfiable instances of projection games.
\end{lemma}

\begin{proof}
  The algorithm works as follows:
  \begin{enumerate}
  \item For each $b$, assign it $\sigma^*_b$ that maximizes $\sum_{a \in N(b)} |\pi_{(a, b)}^{-1}(\sigma_b)|$.
  \item For each $a$, assign it $\sigma^*_a$ that maximizes the number of edges satisfied, $|\{b \in N(a) \mid \pi_{(a, b)}(\sigma_a) = \sigma^*_b\}|$.
    \end{enumerate}

  Let $e^*$ be the number of edges that get satisfied by this algorithm. We have
  \begin{align*}
    e^* &= \sum_{a \in A} |\{b \in N(a) \mid \pi_{(a, b)}(\sigma^*_a) = \sigma^*_b\}|.
  \end{align*}

  Due to the second step, for each $a \in A$, the number of edges satisfied is at least an average of the number of edges satisfy over all assignments in $\Sigma_A$. This can be written as follows.
  \begin{align*}
    e^* &\geq \sum_{a \in A} \frac{\sum_{\sigma_a \in \Sigma_A} |\{b \in N(a) \mid \pi_{(a, b)}(\sigma_a) = \sigma^*_b\}|}{|\Sigma_A|} \\
    &= \sum_{a \in A} \frac{\sum_{b \in N(a)} |\pi^{-1}_{(a, b)}(\sigma^*_b)|}{|\Sigma_A|} \\
    &= \frac{1}{|\Sigma_A|} \sum_{a \in A} \sum_{b \in N(a)} |\pi^{-1}_{(a, b)}(\sigma^*_b)| \\
    &= \frac{1}{|\Sigma_A|} \sum_{b \in B} \sum_{a \in N(b)} |\pi^{-1}_{(a, b)}(\sigma^*_b)|.
  \end{align*}

  Moreover, from the first step, we can conclude that, for each $b$, $\sum_{a \in N(b)} |\pi_{(a, b)}^{-1}(\sigma^*_b)| \geq \sum_{a \in N(b)} |\pi_{(a, b)}^{-1}(\sigma^{OPT}_b)|$. As a result, we can conclude that
  \begin{align*}
    e^* &\geq \frac{1}{|\Sigma_A|} \sum_{b \in B} \sum_{a \in N(b)} |\pi^{-1}_{(a, b)}(\sigma^{OPT}_b)|\\
    & = \frac{1}{|\Sigma_A|}\sum_{e\in E}p_e\\
    &= \frac{|E|\overline{p}}{|\Sigma_A|}
  \end{align*}

  Hence, this algorithm satisfies at least $\frac{\overline{p}}{|\Sigma_A|}$ fraction of the edges. Thus, this is a polynomial-time $\frac{|\Sigma_A|}{\overline{p}}$-approximation algorithm for satisfiable instances of projection games, which concludes our proof.
\end{proof}

\subsection*{Know Your Neighbors' Neighbors Algorithm} \label{subsec-knowyourneighborsneighbors}

The next algorithm shows that if the neighbors of neighbors of a vertex $a_0\in A$ expand, then one can satisfy many of the (many!) edges that touch the neighbors of $a_0$'s neighbors.

\begin{lemma} \label{reduction}
  For each $a_0 \in A$, there exists a polynomial-time $\frac{|E|\overline{p}}{h(a_0)}$-approximation algorithm for satisfiable instances of projection games.
\end{lemma}

\begin{proof}
  To prove Lemma~\ref{reduction}, we want to find an algorithm that satisfies at least $\frac{h(a_0)}{\overline{p}}$ edges for each $a_0 \in A$.

  The algorithm works as follows:
  \begin{enumerate}
  \item Iterate over all assignments $\sigma_{a_0} \in \Sigma_A$ to $a_0$:
    \begin{enumerate}
    \item Assign $\sigma_b = \pi_{(a_0, b)}(\sigma_{a_0})$ to $b$ for all $b \in N(a_0)$.
    \item For each $a \in A$, find the set of plausible assignments to $a$, i.e., $S_a = \{\sigma_a\in\Sigma_A \mid \forall b \in N(a) \cap N(a_0),  \pi_{(a, b)}(\sigma_a) =  \sigma_b\}$. If for any $a$, the set $S_a$ is empty, then we proceed to the next assignment without executing the following steps.
    \item For all $b \in B$, pick an assignment $\sigma^*_b$ for $b$ that maximizes the average number of satisfied edges over all assignments in $S_a$ to vertices $a$ in $N(b) \cap N_2(a_0)$, i.e., maximizes $\sum_{a \in N(b) \cap N_2(a_0)} |\pi^{-1}_{(a, b)}(\sigma_b) \cap S_a|$.
    \item For each vertex $a \in A$, pick an assignment $\sigma^*_a \in S_a$ that maximizes the number of satisfied edges, $|\{b \in N(a) \mid \pi_{(a, b)}(\sigma_a) = \sigma_b^*\}|$.
    \end{enumerate}
  \item Pick an assignment $\{\sigma^*_{a}\}_{a\in A}$, $\{\sigma^*_b\}_{b\in B}$ from the previous step that satisfies the most edges.
  \end{enumerate}

  We will prove that this algorithm indeed satisfies at least $\frac{h(a)}{\overline{p}}$ edges.

  Let $e^*$ be the number of edges satisfied by the algorithm.  We have
  \begin{align*}
    e^* &= \sum_{a \in A} |\{b \in N(a) \mid \pi_{(a, b)}(\sigma^*_a) = \sigma^*_b\}|.
  \end{align*}

  Since in step 1, we try every possible $\sigma_{a_0} \in \Sigma_A$, we must have tried $\sigma_{a_0} = \sigma^{OPT}_{a_0}$. This means that the assignment to $a_0$ is the optimal assignment. As a result, the assignments to every node in $N(a_0)$ is the optimal assignment; that is $\sigma_{b} = \sigma^{OPT}_b$ for all $b \in N(a_0)$. Note that when the optimal assignment is assigned to $a_0$, we have $\sigma^{OPT}_a \in S_a$ for all $a \in A$. This means that the algorithm proceeds until the end. Thus, the solution this algorithm gives satisfies at least as many edges as when $\sigma_v = \sigma^{OPT}_v$ for all $v \in \{a_0\} \cup N(a_0)$. From now on, we will consider only this case.

  Since for each $a \in A$, the assignment $\sigma^*_a$ is chosen to maximize the number of edges satisfied, we can conclude that the number of edges satisfied by selecting $\sigma^*_a$ is at least the average of the number of edges satisfied over all $\sigma_a \in S_a$.

  As a result, we can conclude that
  \begin{align*}
    e^* &\geq \sum_{a \in A} \frac{\sum_{\sigma_a \in S_a} |\{b \in N(a) \mid \pi_{(a, b)}(\sigma_a) = \sigma^*_b\}|}{|S_a|} \\
    &= \sum_{a \in A} \frac{\sum_{\sigma_a \in S_a} \sum_{b \in N(a)} 1_{\pi_{(a, b)}(\sigma_a) = \sigma^*_b}}{|S_a|} \\
    &= \sum_{a \in A} \frac{\sum_{b \in N(a)} \sum_{\sigma_a \in S_a} 1_{\pi_{(a, b)}(\sigma_a) = \sigma^*_b}}{|S_a|} \\
    &= \sum_{a \in A} \frac{\sum_{b \in N(a)} |\pi_{(a, b)}^{-1}(\sigma^*_b) \cap S_a|}{|S_a|} \\
    &= \sum_{b \in B} \sum_{a \in N(b)} \frac{|\pi_{(a, b)}^{-1}(\sigma^*_b) \cap S_a|}{|S_a|} \\
    &\geq \sum_{b \in B} \sum_{a \in N(b) \cap N_2(a_0)} \frac{ |\pi_{(a, b)}^{-1}(\sigma^*_b) \cap S_a|}{|S_a|}
  \end{align*}

  Now, for each $a \in N_2(a_0)$, consider $S_a$. From the definition of $S_a$, we have
  \begin{align*}
    S_a =  \{\sigma_a \in \Sigma_A\mid \forall b \in N(a) \cap N(a_0),  \pi_{(a, b)}(\sigma_a) =  \sigma_b\}
    = \bigcap_{b \in N(a) \cap N(a_0)} \pi^{-1}_{(a, b)}(\sigma_b).
  \end{align*}

  As a result, we can conclude that
  \begin{align*}
    |S_a| &\leq \min_{b \in N(a) \cap N(a_0)}\{|\pi^{-1}_{(a, b)}(\sigma_b)|\} \\
    &= \min_{b \in N(a) \cap N(a_0)}\{|\pi^{-1}_{(a, b)}(\sigma^{OPT}_b)|\} \\
    &= \min_{b \in N(a) \cap N(a_0)}\{p_{(a, b)}\}.
  \end{align*}

  Note that since $a \in N_2(a_0)$, we have $N(a) \cap N(a_0) \ne \emptyset$. Since we assume for simplicity that $p_e = \overline{p}$ for all $e \in E$, we can conclude that $|S_a| \leq \overline{p}$.

  This implies that
  \begin{align*}
    e^* \geq \frac{1}{\overline{p}} \sum_{b \in B} \sum_{a \in N(b) \cap N_2(a_0)} |\pi_{(a, b)}^{-1}(\sigma^*_b) \cap S_a|.
  \end{align*}

  Since we pick the assignment $\sigma^*_b$ that maximizes $\sum_{a \in N(b) \cap N_2(a_0)} |\pi^{-1}_{(a, b)}(\sigma^*_b) \cap S_a|$ for each $b \in B$, we can conclude that
  \begin{align*}
    e^* &\geq \frac{1}{\overline{p}} \sum_{b \in B} \sum_{a \in N(b) \cap N_2(a_0)} |\pi_{(a, b)}^{-1}(\sigma^*_b) \cap S_a|\\
    &\geq \frac{1}{\overline{p}} \sum_{b \in B} \sum_{a \in N(b) \cap N_2(a_0)} |\pi_{(a, b)}^{-1}(\sigma^{OPT}_b) \cap S_a|.
  \end{align*}

  Since the optimal assignment satisfies every edge, we can conclude that $\sigma^{OPT}_a \in \pi_{(a, b)}^{-1}(\sigma^{OPT}_b)$ and $\sigma^{OPT}_a \in S_a$, for all $b \in B$ and $a \in N(b) \cap N_2(a_0)$. This implies that
  \begin{align*}
    e^* &\geq \frac{1}{\overline{p}} \sum_{b \in B} \sum_{a \in N(b) \cap N_2(a_0)} |\pi_{(a, b)}^{-1}(\sigma^{OPT}_b) \cap S_a| \\
    &\geq \frac{1}{\overline{p}} \sum_{b \in B} \sum_{a \in N(b) \cap N_2(a_0)} 1.
  \end{align*}

  The last term can be written as
  \begin{align*}
    \frac{1}{\overline{p}} \sum_{b \in B} \sum_{a \in N(b) \cap N_2(a_0)} 1 &= \frac{1}{\overline{p}} \sum_{a \in N_2(a_0)} \sum_{b \in N(a)} 1 \\
    &= \frac{1}{\overline{p}} (h(a_0)) \\
    &= \frac{h(a_0)}{\overline{p}}.
  \end{align*}

  As a result, we can conclude that this algorithm gives an assignment that satisfies at least $\frac{h(a_0)}{\overline{p}}$ edges out of all the $|E|$ edges. Hence, this is a polynomial-time $\frac{|E|\overline{p}}{h(a_0)}$-approximation algorithm as desired.
\end{proof}

\subsubsection{Divide and Conquer Algorithm.}

We will present an algorithm that separates the graph into disjoint subgraphs for which we can find the optimal assignments in polynomial time. We shall show below that, if $h(a)$ is small for all $a \in A$, then we are able to find such subgraphs that contain most of the graph's edges.

\begin{lemma} \label{separation} There exists a polynomial-time $O\left(\frac{n_An_Bh_{max}}{|E|^2}\right)$-approximation algorithm for satisfiable instances of projection games.
\end{lemma}

\begin{proof}
  To prove Lemma~\ref{separation}, we will describe an algorithm that gives an assignment that satisfies $\Omega\left(\frac{|E|^3}{n_An_Bh_{max}}\right)$ edges.

  We use $\mathcal{P}$ to represent the collection of subgraphs we find. The family $\mathcal{P}$ consists of disjoint sets of vertices. Let $V_\mathcal{P}$ be $\bigcup_{P \in \mathcal{P}} P$.

  For any set $S$ of vertices, define $G_S$ to be the graph induced on $S$ with respect to $G$. Moreover, define $E_S$ to be the set of edges of $G_S$. We also define $E_\mathcal{P} = \bigcup_{P \in \mathcal{P}} E_P$.

  The algorithm works as follows.
  \begin{enumerate}
  \item Set $\mathcal{P} \leftarrow \emptyset$.
  \item\label{a:while} While there exists a vertex $a \in A$ such that $|E_{(N(a) \cup N_2(a)) - V_\mathcal{P}}| \geq \frac{1}{4} \frac{|E|^2}{n_An_B}$:
    \begin{enumerate}
    \item\label{a:update-P} Set $\mathcal{P} \leftarrow \mathcal{P} \cup \{(N_2(a) \cup N(a)) - V_\mathcal{P}\}$.

    \end{enumerate}
  \item\label{a:assign} For each $P \in \mathcal{P}$, find in time $poly(|\Sigma_A|, |P|)$ an assignment to the vertices in $P$ that satisfies all the edges spanned by $P$.
\end{enumerate}

  We will divide the proof into two parts. First, we will show that when we cannot find a vertex $a$ in step~\ref{a:while}, $\left|E_{(A \cup B) - V_\mathcal{P}} \right| \leq \frac{|E|}{2}$. Second, we will show that the resulting assignment from this algorithm satisfies $\Omega\left(\frac{|E|^3}{n_An_Bh_{max}}\right)$ edges.

  We will start by showing that if no vertex $a$ in step~\ref{a:while} exists, then $\left|E_{(A \cup B) - V_\mathcal{P}} \right| \leq \frac{|E|}{2}$.

  Suppose that we cannot find a vertex $a$ in step~\ref{a:while}. In other words, $|E_{(N(a) \cup N_2(a)) - V_\mathcal{P}}| < \frac{1}{4} \frac{|E|^2}{n_An_B}$ for all $a \in A$.

  Consider $\sum_{a \in A} |E_{(N(a) \cup N_2(a)) - V_\mathcal{P}}|$. Since $|E_{(N(a) \cup N_2(a)) - V_\mathcal{P}}| < \frac{1}{4} \frac{|E|^2}{n_An_B}$ for all $a \in A$, we have the following inequality.
  \begin{align*}
    \frac{|E|^2}{4n_B} \geq \sum_{a \in A} |E_{(N(a) \cup N_2(a)) - V_\mathcal{P}}|.
  \end{align*}


  Let $N^p(v) = N(v) - V_\mathcal{P}$ and $N_2^p(v) = N_2(v) - V_\mathcal{P}$. Similary, define $N^p(S)$ for a subset $S\subseteq A\cup B$. It is easy to see that $N_2^p(v) \supseteq N^p(N^p(v))$. This implies that, for all $a \in A$, we have $|E_{(N^p(a) \cup N^p_2(a))}| \geq |E_{(N^p(a) \cup N^p(N^p(a)))}|$. Moreover, it is not hard to see that, for all $a \in A - V_\mathcal{P}$, we have $|E_{(N^p(a) \cup N^p(N^p(a)))}| = \sum_{b \in N^p(a)} |N^p(b)|$.

  Thus, we can derive the following:
  \begin{align*}
    \sum_{a \in A} |E_{(N(a) \cup N_2(a)) - V_\mathcal{P}}| &= \sum_{a \in A} |E_{(N^p(a) \cup N^p_2(a))}| \\
    &\geq \sum_{a \in A - V_\mathcal{P}} |E_{(N^p(a) \cup N^p_2(a))}| \\
    &\geq \sum_{a \in A - V_\mathcal{P}} \sum_{b \in N^p(a)} |N^p(b)| \\
    &= \sum_{b \in B - V_\mathcal{P}} \sum_{a \in N^p(b)} |N^p(b)| \\
    &= \sum_{b \in B - V_\mathcal{P}} |N^p(b)|^2. \\
  \end{align*}

  From Jensen's inequality, we have
  \begin{align*}
    \sum_{a \in A} |E_{(N(a) \cup N_2(a)) - V_\mathcal{P}}| &\geq \frac{1}{|B - V_\mathcal{P}|} \left(\sum_{b \in B - V_\mathcal{P}} |N^p(b)| \right)^2 \\
    &= \frac{1}{|B - V_\mathcal{P}|} \left|E_{\left(A \cup B\right) - V_\mathcal{P}}\right|^2 \\
    & \geq \frac{1}{n_B} \left|E_{\left(A \cup B\right) - V_\mathcal{P}}\right|^2.
  \end{align*}

  Since $\frac{|E|^2}{4n_B} \geq \sum_{a \in A} |E_{(N(a) \cup N_2(a)) - V_\mathcal{P}}|$ and $\sum_{a \in A} |E_{(N(a) \cup N_2(a)) - V_\mathcal{P}}| \geq  \frac{1}{n_B} \left|E_{\left(A \cup B\right) - V_\mathcal{P}}\right|^2$, we can conclude that
  \begin{align*}
    \frac{|E|}{2} \geq \left|E_{\left(A \cup B\right) - V_\mathcal{P}}\right|
  \end{align*}
  which concludes the first part of the proof.

  Next, we will show that the assignment the algorithm finds satisfies at least $\Omega\left(\frac{|E|^3}{n_An_Bh_{max}}\right)$ edges. Since we showed that $\frac{|E|}{2} \geq \left|E_{\left(A \cup B\right) - V_\mathcal{P}}\right|$ when the algorithm terminates, it is enough to prove that $|E_\mathcal{P}| \geq \frac{|E|^2}{4n_An_Bh_{max}} \left(|E| - \left|E_{\left(A \cup B\right) - V_\mathcal{P}}\right|\right)$. Note that the algorithm guarantees to satisfy all the edges in $E_\mathcal{P}$.

  We will prove this by using induction to show that at any point in the algorithm, $|E_\mathcal{P}| \geq \frac{|E|^2}{4n_An_Bh_{max}} \left(|E| - \left|E_{\left(A \cup B\right) - V_\mathcal{P}}\right|\right)$.

  \emph{Base Case.} At the beginning, we have $|E_\mathcal{P}| = 0 = \frac{|E|^2}{4n_An_Bh_{max}} \left(|E| - \left|E_{\left(A \cup B\right) - V_\mathcal{P}}\right|\right)$, which satisfies the inequality.

  \emph{Inductive Step.} The only step in the algorithm where any term in the inequality changes is step~\ref{a:update-P}. Let $\mathcal{P}_{old}$ and $\mathcal{P}_{new}$ be the set $\mathcal{P}$ before and after step~\ref{a:update-P} is executed, respectively. Let $a$ be the vertex selected from step~\ref{a:while}. Suppose that $\mathcal{P}_{old}$ satisfies the inequality.

  From the condition in step~\ref{a:while}, we have $|E_{(N(a) \cup N_2(a)) - V_{\mathcal{P}_{old}}}| \geq \frac{1}{4} \frac{|E|^2}{n_An_B}$. Since $|E_{\mathcal{P}_{new}}| = |E_{\mathcal{P}_{old}}| + |E_{(N(a) \cup N_2(a)) - V_{\mathcal{P}_{old}}}|$, we have
  \begin{align*}
    |E_{\mathcal{P}_{new}}| \geq |E_{\mathcal{P}_{old}}| + \frac{1}{4} \frac{|E|^2}{n_An_B}.
  \end{align*}

  Now, consider $\left(|E| - |E_{\left(A \cup B\right) - V_{\mathcal{P}_{new}}}|\right) - \left(|E| - |E_{\left(A \cup B\right) - V_{\mathcal{P}_{old}}}|\right)$. We have
  \begin{align*}
    \left(|E| - |E_{\left(A \cup B\right) - V_{\mathcal{P}_{new}}}|\right) - \left(|E| - |E_{\left(A \cup B\right) - V_{\mathcal{P}_{old}}}|\right) = |E_{\left(A \cup B\right) - V_{\mathcal{P}_{old}}}| - |E_{\left(A \cup B\right) - V_{\mathcal{P}_{new}}}|
  \end{align*}

  Since $V_{\mathcal{P}_{new}} = V_{\mathcal{P}_{old}} \cup \left(N_2(a) \cup N(a)\right)$, we can conclude that
  \begin{align*}
    \left((A \cup B) - V_{\mathcal{P}_{old}}\right) \subseteq \left((A \cup B) - V_{\mathcal{P}_{new}}\right) \cup \left(N_2(a) \cup N(a)\right).
  \end{align*}

  Thus, we can also derive
  \begin{align*}
    E_{(A \cup B) - V_{\mathcal{P}_{old}}} &\subseteq E_{\left((A \cup B) - V_{\mathcal{P}_{new}}\right) \cup \left(N_2(a) \cup N(a)\right)} \\
    &= E_{(A \cup B) - V_{\mathcal{P}_{new}}} \cup \{(a', b') \in E \mid a' \in N_2(a) \text{ or } b' \in N(a)\}. \\
  \end{align*}

  From the definition of $N$ and $N_2$, for any $(a', b') \in E$, if $b' \in N(a)$ then $a' \in N_2(a)$. Thus, we have $\{(a', b') \in E \mid a' \in N_2(a) \text{ or } b' \in N(a)\} = \{(a', b') \in E \mid a' \in N_2(a)\} = E(N_2(a))$. The cardinality of the last term was defined to be $h(a)$. Hence, we can conclude that
  \begin{align*}
    |E_{(A \cup B) - V_{\mathcal{P}_{old}}}| &\leq |E_{(A \cup B) - V_{\mathcal{P}_{new}}} \cup \{(a', b') \in E \mid a' \in N_2(a) \text{ or } b' \in N(a)\}| \\
    &\leq |E_{(A \cup B) - V_{\mathcal{P}_{new}}}| + |\{(a', b') \in E \mid a' \in N_2(a) \text{ or } b' \in N(a)\}| \\
    &= |E_{(A \cup B) - V_{\mathcal{P}_{new}}}| + |\{(a', b') \in E \mid a' \in N_2(a)\}| \\
    &= |E_{(A \cup B) - V_{\mathcal{P}_{new}}}| + |E(N_2(a))| \\
    &= |E_{(A \cup B) - V_{\mathcal{P}_{new}}}| + h(a) \\
    &\leq |E_{(A \cup B) - V_{\mathcal{P}_{new}}}| + h_{max}.
  \end{align*}

  This implies that $\left(|E| - \left|E_{\left(A \cup B\right) - V_{\mathcal{P}}}\right|\right)$ increases by at most $h_{max}$.

  Hence, since $\left(|E| - \left|E_{\left(A \cup B\right) - V_{\mathcal{P}}}\right|\right)$ increases by at most $h_{max}$ and $\left|E_\mathcal{P}\right|$ increases by at least $\frac{1}{4}\frac{|E|^2}{n_An_B}$ and from the inductive hypothesis, we can conclude that
  \begin{align*}
    |E_{\mathcal{P}_{new}}| \geq \frac{|E|^2}{4n_An_Bh_{max}} \left(|E| - \left|E_{\left(A \cup B\right) - V_{\mathcal{P}_{new}}}\right|\right).
  \end{align*}

  Thus, the inductive step is true and the inequality holds at any point during the execution of the algorithm.

  When the algorithm terminates, since $|E_\mathcal{P}| \geq \frac{|E|^2}{4n_An_Bh_{max}} \left(|E| - \left|E_{\left(A \cup B\right) - V_\mathcal{P}}\right|\right)$ and $\frac{|E|}{2} \geq \left|E_{\left(A \cup B\right) - V_\mathcal{P}}\right|$, we can conclude that $|E_{\mathcal{P}}| \geq \frac{|E|^3}{8n_An_Bh_{max}}$. Since the algorithm guarantees to satisfy every edge in $E_{\mathcal{P}}$, we can conclude that the algorithm gives $O(\frac{n_An_Bh_{max}}{|E|^2})$ approximation ratio, which concludes our proof of Lemma~\ref{separation}.
  \end{proof}

\subsection*{Proof of Theorem~\ref{t:approx}}


\begin{proof}
  Using Lemma \ref{reduction} with $a_0$ that maximizes the value of $h(a_0)$, i.e., $h(a_0) = h_{max}$, we can conclude that there exists a polynomial-time $\frac{|E|\overline{p}}{h_{max}}$-approximation algorithm for satisfiable instances of projection games.

  Moreover, from Lemmas~\ref{dBapprox},~\ref{pickbest} and~\ref{separation}, there exists a polynomial-time $\frac{|E|}{n_B}$-approximation algorithm, a polynomial-time $\frac{|\Sigma_A|}{\overline{p}}$-approximation algorithm and a polynomial time $O\left(\frac{n_An_Bh_{max}}{|E|^2}\right)$-approximation algorithm for satisfiable instances of projection games.

  By picking the best out of these four algorithms, we can get an approximation ratio of $O\left(\min\left(\frac{|E|\overline{p}}{h_{max}}, \frac{|\Sigma_A|}{\overline{p}}, \frac{|E|}{n_B}, \frac{n_An_Bh_{max}}{|E|^2}\right)\right)$.

  Since the minimum is at most the value of the geometric mean, we deduce that the approximation ratio is
  \begin{align*}
    O\left(\sqrt[4]{\frac{|E|\overline{p}}{h_{max}} \cdot \frac{|\Sigma_A|}{\overline{p}} \cdot \frac{|E|}{n_B} \cdot \frac{n_An_Bh_{max}}{|E|^2}}\right) = O\left(\sqrt[4]{n_A|\Sigma_A|}\right).
  \end{align*}
  This concludes the proof of Theorem~\ref{t:approx}.
\end{proof}



\section{Sub-Exponential Time Algorithms for Smooth Projection Games}
In this section, we prove Theorem~\ref{t:smooth} via Lemma~\ref{lemmasmoothexact} and Lemma~\ref{lemmasmoothapprox} stated in the following subsections.

Before we proceed to state and prove the lemmas, let us start by noting the following observation which follows immediately from the definition of smoothness (Definition~\ref{d:smooth}).

\begin{observation} \label{o:smooth}
	For any $\mu$-smooth satisfiable projection game and any vertex $a$ in $A$, if more than $\mu d_a$ neighbors of $a$ are assigned according to the optimal assignment, then there is only one $\sigma_a \in \Sigma_A$ that satisfies all edges from $a$ to those neighbors. This $\sigma_a$ is the optimal assignment for $a$.
\end{observation}

We will use this observation in the proofs of both lemmas.

\subsection{Exact Algorithm for Graphs With Sufficiently Large Degrees}

The idea of this algorithm is to randomly select $\Theta(\mu n_B)$ vertices from $B$ and try all possible assignments for them. When the assignment for the selected set is correct, we can determine the optimal assignment for every $a \in A$ such that more than $\mu d_a$ of its neighbors are in the selected set.

The next lemma shows that, provided that the degrees of the vertices in $A$ are not too small, the algorithm gives an assignment that satisfies all the edges with high probability.

\begin{lemma} \label{lemmasmoothexact} For every constant $c\geq 1$, the following statement holds:
  given a satisfiable instance of projection games that satisfies the $\mu$-smoothness property and that $d_a \geq \frac{c\log n_A}{\mu}$ for all $a \in A$, one can find the optimal assignment for the game in time $\exp(O(\mu n_B \log |\Sigma_B|))\cdot poly(n_A,\Sigma_A)$ with probability $1/2$.
\end{lemma}

\begin{proof}
  Let $c_1$ be a constant greater than one.

  The algorithm is as follows.
  \begin{enumerate}
  \item For each $b \in B$, randomly pick it with probability $c_1 \mu$. Call the set of all picked vertices $B^*$.
  \item Try every possible assignments for the vertices in $B^*$. For each assignment:
    \begin{enumerate}
    \item For each node $a \in A$, try every assignment $\sigma_a\in\Sigma_A$ for it. If there is exactly one assignment that satisfies all the edges that touch $a$, pick that assignment.
    \end{enumerate}
  \item If encountered an assignment satisfying all edges, return that assignment.
  \end{enumerate}

  Next, we will show that, with probability $1/2$, the algorithm returns the optimal assignment in time $\exp(O(\mu n_B \log |\Sigma_B|))\cdot poly(n_A,|\Sigma_A|)$.

  For each $b \in B$, let $X_b$ be an indicator variable for whether the vertex $b$ is picked, i.e. $b \in B^*$. From step 1, we have
  \begin{align*}
    E[X_b] = c_1 \mu.
  \end{align*}

  Let $X$ be a random variable representing the number of vertices that are selected in step 1, i.e. $X = |B^*|$. We have
  \begin{align*}
    E[X] &= \sum_{b \in B} E[X_b] = n_B c_1 \mu.
  \end{align*}

  For each $a \in A$, let $Y_a$ be a random variable representing the number of their neighbors that are picked, i.e. $Y_a = |B^* \cap N(a)|$. We have
  \begin{align*}
    E[Y_a] &= \sum_{b \in N(a)} E[X_b] = d_a c_1 \mu.
  \end{align*}

  Clearly, by iterating over all possible assignments for $B^*$, the algorithm running time is $|\Sigma_B|^{O(|B^*|)} = \exp(O(|B^*|\log |\Sigma_B|))$. Thus, if $X = |B^*| = O(\mu n_B)$, the running time for the algorithm is $\exp(O(\mu n_B \log |\Sigma_B|))$.

  Let $c_2$ be a constant greater than one. If $X \leq c_2c_1 \mu n_B$, then the running time of the algorithm is $\exp(O(\mu n_B \log |\Sigma_B|))$ as desired.

  Since $\{X_b\}_{b \in B}$ are independent, using Chernoff bound, we have
  \begin{align*}
    Pr[X > c_2 c_1 \mu n_B] &= Pr[X > c_2E[X]] \\
    &< \left(\frac{e^{c_2 - 1}}{(c_2)^{c_2}}\right)^{n_B c_1 \mu} \\
    &= e^{-n_Bc_1(c_2(\log c_2 - 1) + 1) \mu}.
  \end{align*}

  Now, consider each $a \in A$. By going through all possible combinations of assignments of vertices in $B^*$, we must go over the optimal assignment for $B^*$. In the optimal assignment, if more than $\mu$ fraction of $a$'s neighbors is in $B^*$ (i.e., $Y_a > \mu d_a$), then in step 3, we get the optimal assignment for $a$ due to Observation~\ref{o:smooth}. Since $\{X_b\}_{b \in N(a)}$ are independent, using Chernoff bound, we can obtain the following inequality.
  \begin{align*}
    Pr[Y_a \leq d_a \mu] &= Pr[Y_a \leq \frac{1}{c_1}E[Y_a]] \\
    &< \left(\frac{e^{\frac{1}{c_1} - 1}}{\left(\frac{1}{c_1}\right)^{\frac{1}{c_1}}}\right)^{d_a c_1 \mu} \\
    &= e^{-d_ac_1(1 - \frac{1}{c_1} - \frac{1}{c_1}\log c_1) \mu}.
  \end{align*}

  Hence, we can conclude that the probability that this algorithm returns an optimal solution within $\exp(O(\mu n_B \log |\Sigma_B|))$-time is at least
  \begin{align*}
    Pr\left[\left(\bigwedge_{a \in A}(Y_a > d_a \mu)\right) \wedge (X \leq c_2 c_1 \mu n_B)\right]
    &= 1 - Pr\left[\left(\bigvee_{a \in A}(Y_a \leq d_a \mu)\right) \vee (X > c_2 c_1 \mu n_B)\right] \\
    &\geq 1 - \left(\sum_{a \in A} Pr[Y_a > d_a \mu]\right) - Pr[X > c_2 c_1 \mu n_B] \\
    &> 1 - \left(\sum_{a \in A} e^{-d_ac_1(1 - \frac{1}{c_1} - \frac{1}{c_1}\log c_1) \mu}\right) - e^{-n_Bc_1(c_2(\log c_2 - 1) + 1) \mu}
  \end{align*}

  Since $c_1(1 - \frac{1}{c_1} - \frac{1}{c_1}\log c_1)$ and $c_1(c_2(\log c_2 - 1) + 1)$ are constant, we can define constants $c_3 = c_1(1 - \frac{1}{c_1} - \frac{1}{c_1}\log c_1)$ and $c_4 = c_1(c_2(\log c_2 - 1) + 1)$. The probability that the algorithm returns an optimal solution can be written as
  \begin{align*}
    1 - e^{-n_B\mu c_4} - \sum_{a \in A} e^{-d_a\mu c_3}.
  \end{align*}

  Moreover, since we assume that $d_a \geq \frac{c\log n_A}{\mu}$ for all $a \in A$, we can conclude that the probability above is at least
  \begin{align*}
    1 - e^{-n_B\mu c_4} - \sum_{a \in A} e^{-c c_3 \log n_A} &= 1 - e^{-n_B\mu c_4} - n_A e^{-c c_3 \log n_A} \\
    &= 1 - e^{-n_B\mu c_4} - e^{- (c c_3 - 1) \log n_A} \\
  \end{align*}

  Note that for any constants $c^*_3, c^*_4$, we can choose constants $c_1, c_2$ so that $c_3 = c_1(1 - \frac{1}{c_1} - \frac{1}{c_1}\log c_1) \geq c^*_3$ and $c_4 = c_1(c_2(\log c_2 - 1) + 1) \geq c^*_4$. 
  This means that we can select $c_1$ and $c_2$ so that $c_3 \geq \frac{1}{c} + \frac{2}{c\log n_A} \in O(1)$ and $c_4 \geq \frac{2}{n_B \mu} \in O(1)$. Note also that here we can assume that $\log n_A > 0$ since an instance is trivial when $n_A = 1$. Plugging $c_3$ and $c_4$ into the lower bound above, we can conclude that, for this $c_1$ and $c_2$, the algorithm gives the optimal solution in the desired running time with probability more than 1/2.
\end{proof}

\subsection{Deterministic Approximation Algorithm For General Degrees}

A deterministic version of the above algorithm is shown below. In this algorithm, we are able to achieve an $O(1)$ approximation ratio within asymptotically the same running time as the algorithm above. In contrast to the previous algorithm, this algorithm works even when the degrees of the vertices are small.

The idea of the algorithm is that, instead of randomly picking a subset $B^*$ of $B$, we will deterministically pick $B^*$. We say that a vertex $a\in A$ is {\em saturated} if more than $\mu$ fraction of its neighbors are in $B^*$, i.e. $|N(a) \cap B^*| > \mu d_a$.
In each step, we pick a vertex in $B$ that neighbors the highest number of unsaturated vertices, and add it to $B^*$. We do this until a constant fraction of the edges are satisfied.

\begin{lemma} \label{lemmasmoothapprox}
  There exists an $\exp(O(\mu n_B \log |\Sigma_B|))\cdot poly(n_A,|\Sigma_A|)$-time $O(1)$-approximation algorithm for satisfiable $\mu$-smooth projection game instances.
\end{lemma}

\begin{proof}
  First, observe that when $\mu \geq 1/4$, we can simply enumerate all the assignments for $B$ and then pick an assignment in $A$ that maximizes number of satisfied edges. This algorithm runs in $\exp(O(n_B \log |\Sigma_B|))\cdot poly(n_A,|\Sigma_A|) = \exp(O(\mu n_B \log |\Sigma_B|))\cdot poly(n_A,|\Sigma_A|)$ time and always returns an optimal assignment as desired.

  From this point onwards, we will focus on only when $\mu < 1/4$. Next, we will consider two cases based on whether $n_A/|E|$ is larger than $1/4$.

  First, if $n_A/|E| \geq 1/4$, then we use the following polynomial-time algorithm that works in a similar manner as the \emph{satisfy one neighbor} algorithm:
  \begin{enumerate}
  \item For each $b \in B$, assign to $b$ any $\sigma_b \in \Sigma_B$ such that $\pi_{(a', b)}^{-1}(\sigma_b) \ne \emptyset$ for every $a' \in N(b)$.
  \item For each $a \in A$, pick one neighbor $b_a$ of $a$. Assign to $a$ any $\sigma_a \in \pi_{(a, b_a)}^{-1}(\sigma_b)$.
  \end{enumerate}
  Observe that, in the first step, we can find $\sigma_b$ for every $b \in B$ because $\sigma_b^{OPT}$ satisfies the condition. Moreover, from the second step, $(a, b_a)$ is satisfied for every $a \in A$. Thus, the output assignment that satisfies at least $n_A$ edges, which is at least $1/4 = \Omega(1)$ fraction of all the edges as desired.

  We now turn our attention to the remaining case where $n_A/|E| < 1/4$. Observe that, without loss of generality, we can assume that $\mu \geq 1/d_a$ for all $a \in A$. This is because $\mu$-smoothness is the same as uniqueness when $\mu < 1/d_a$, and, after determining an optimal assignment for all saturated vertices, one can also find an optimal assignment for all the $a$'s with $\mu <1/d_a$ in polynomial time (similarly to solving satisfiable instances of unique games in polynomial time).

  For convenience, let $c_1$ be $1/4$. The algorithm for this case can be described in the following steps.
  \begin{enumerate}
  \item Set $B^* \leftarrow \emptyset$.
  \item Let $S$ be the set of all saturated vertices, i.e. $S = \{a \in A \mid |N(a) \cap B^*| > \mu d_a\}$. As long as $|\sum_{a \in S} d_a| < c_1 |E|$:
    \begin{enumerate}
    \item Pick a vertex $b^* \in B - B^*$ with maximal $|N(b) - S|$. Set $B^* \leftarrow B^* \cup \{b^*\}$.

    \end{enumerate}
  \item Iterate over all possible assignments to the vertices in $B^*$:
    \begin{enumerate}
    \item For each saturated vertex $a \in S$, search for an assignment that satisfies all edges in $\{a\}\times (N(a) \cap B^*)$. If, for any saturated vertex $a \in S$, this assignment cannot be found, skip the next part and go to the next assignment for $B^*$.

    \item Assign each vertex in $B$ an assignment in $\Sigma_B$ that  satisfies the maximum number of edges that touch it.
    \item Assign arbitrary elements from $\Sigma_A$ to the vertices in $A$ that are not yet assigned.
    \end{enumerate}
  \item Output the assignment that satisfies the maximum number of edges.
  \end{enumerate}

  We will divide the proof into two steps. First, we will show that the number of edges satisfied by the output assignment is at least $c_1|E|$. Second, we will show that the running time for this algorithm is $\exp(O(\mu n_B \log |\Sigma_B|))\cdot poly(n_A,|\Sigma_A|)$.

  Observe that since we are going through all the possible assignments of $B^*$, we must go through the optimal assignment. Focus on this assignment. From Observation~\ref{o:smooth}, for each saturated vertex $a \in S$, there is exactly one assignment that satisfies all the edges to $B^*$; this assignment is the optimal assignment. Since we have the optimal assignments for all $a \in S$, we can satisfy all the edges with one endpoint in $S$; the output assignment satisfies at least $\sum_{a \in S} d_a$ edges. Moreover, when the algorithm terminates, the condition in step 2 must be false. Thus, we can conclude that $\sum_{a \in S} d_a \geq c_1|E|$. As a result, the algorithm gives an approximation ratio of $\frac{1}{c_1} = 4 = O(1)$.

  Since we go through all possible assignments for $B^*$, the running time of this algorithm is $|\Sigma_B|^{O(|B^*|)}\cdot poly(n_A,|\Sigma_A|)$. In order to prove that the running time is $\exp(O(\mu n_B \log |\Sigma_B|))\cdot poly(n_A,|\Sigma_A|)$, we only need to show that $|B^*| = O(\mu n_B)$.

  When the algorithm picks a vertex $b^*\in B$ to $B^*$ we say that it {\em hits} all its neighbors that are unsaturated.
  Consider the total number of hits to all the vertices in $A$. Since saturated vertices do not get hit any more, we can conclude that each vertex $a \in A$ gets at most $\mu d_a + 1$ hits. As a result, the total number of hits to all vertices $a \in A$ is at most
  \begin{align*}
    \sum_{a \in A} (\mu d_a + 1) &= \mu |E| + n_A.
  \end{align*}

  Next, consider the set $B^*$. Let $B^* = \{b_1, \dots, b_m\}$ where $b_1, \cdots, b_m$ are sorted by the time, from the earliest to the latest, they get added into $B^*$. Let $v(b_i)$ be the number of hits $b_i$ makes. Since the total number of hits by $B^*$ equals the total number of hits to $A$, from the bound we established above, we have
  \begin{align*}
    \sum_{i=1}^{m} v(b_i) \leq \mu |E| + n_A.
  \end{align*}

  Now, consider the adding of $b_i$ to $B^*$. Let $B^*_{i}$ be $\{b_1, \dots, b_{i-1}\}$, the set $B^*$ at the time right before $b_i$ is added to $B^*$, and let $S_i$ be $\{a \in A \mid |N(a) \cap B^*_i| > \mu d_a\}$, the set of saturated vertices at the time right before $b_i$ is added to $B^*$. Since we are picking $b_i$ from $B - B^*_i$ with the maximum number of hits, the number of hits from $b_i$ is at least the average number of possible hits over all vertices in $B - B^*_i$. That is
  \begin{align*}
    v(b_i) &= |N(b_i) - S_i| \\
    &\geq \frac{1}{|B - B^*_i|}\left(\sum_{b \in B - B^*_i} |N(b) - S_i|\right) \\
    &\geq \frac{1}{n_B}\left(\left(\sum_{b \in B} |N(b) - S_i|\right) - \left(\sum_{b \in B^*_i} |N(b) - S_i|\right)\right).
  \end{align*}
  We can also derive the following inequality.
  \begin{align*}
    \sum_{b \in B} |N(b) - S_i| &= \sum_{b \in B} \sum_{a \in N(b) - S_i} 1 \\
    &=\sum_{b \in B} \sum_{a \in A - S_i} \textbf{1}_{(a, b) \in E} \\
    &= \sum_{a \in A - S_i} \sum_{b \in B}  \textbf{1}_{(a, b) \in E} \\
    &= \sum_{a \in A - S_i} d_a \\
    &= |E| - \left(\sum_{a \in S_i} d_a\right) \\
    &> (1 - c_1)|E|.
  \end{align*}
  Note that the last inequality comes from the condition in step 2 of the algorithm.

  Moreover, we have
  \begin{align*}
    \sum_{b \in B^*_i} |N(b) - S_i| &= \sum_{j = 1}^{i-1} |N(b_j) - S_i| \\
      (\text{Since } S_j \subseteq S_i) &\leq \sum_{j = 1}^{i-1} |N(b_j) - S_j| \\
    &= \sum_{j=1}^{i-1} v(b_j) \\
    &\leq \sum_{j=1}^{m} v(b_j) \\
    &\leq \mu|E| + n_A.
  \end{align*}

  Putting them together, we have
  \begin{align*}
    v(b_i) > \frac{1}{n_B} ((1 - c_1)|E| - \mu|E| - n_A)
  \end{align*}
  for all $i = 1, \dots, m$

  From this and from $\sum_{i=1}^{m} v(b_i) \leq \mu|E| + n_A$, we can conclude that
  \begin{align*}
    m &< \frac{\mu|E| + n_A}{\frac{1}{n_B} ((1 - c_1)|E| - \mu|E| - n_A)} \\
    &= n_B \mu \left(\frac{1 + \frac{n_A}{|E|\mu}}{(1 - c_1) - \mu - \frac{n_A}{|E|}}\right).
  \end{align*}
  Consider the term $\frac{1 + \frac{n_A}{|E|\mu}}{(1 - c_1) - \mu - \frac{n_A}{|E|}}$. Since $c_1 = 1/4, \mu < 1/4$ and $\frac{n_A}{|E|} < 1/4$, we can conclude that the denominator is $\Theta(1)$. \\

  Consider $\frac{n_A}{|E|\mu}$. We have
    \begin{align*}
      |E|\mu &= \sum_{a \in A} \mu d_a.
    \end{align*}
    Since we assume that $d_a\geq 1/\mu$ for all $a\in A$, we have $n_A \leq |E|\mu$. In other words, $\frac{n_A}{|E|\mu} \leq 1$.
    Hence, we can conclude that the numerator is $\Theta(1)$.

    As a result, we can deduce that $m = O(\mu n_B)$. Thus, the running time for the algorithm is $\exp(O(\mu n \log |\Sigma_B|))\cdot poly(n_A,|\Sigma_A|)$, which concludes our proof.
\end{proof}

\section{PTAS for Projection Games on Planar Graphs} \label{s:planar}

In this section, we prove our results for projection games on planar graphs. The section is divided naturally into three subsections. First, we show the NP-hardness of projection games on planar graphs. Next, we describe PTAS for the problem. Lastly, we prove a matching running time lower bound for the PTAS.

\subsection{NP-Hardness of Projection Games on Planar Graphs}

The NP-hardness of projection games on planar graphs is proved by reduction from 3-coloring on planar graphs. The latter was proven to be NP-hard by Garey, Johnson and Stockmeyer~\cite{GJS}.

\begin{theorem}
  \textsc{Label Cover} on planar graphs is NP-hard.
\end{theorem}
\begin{proof}
  We will prove this by reducing from 3-colorability problem on planar graph, which was proven by Garey, Johnson and Stockmeyer to be NP-hard~\cite{GJS}. The problem can be formally stated as following.

  \textsc{Planar Graph 3-Colorability: }
  Given a planar graph $\check{G} = (\check{V}, \check{E})$, decide whether it is possible to assign each node a color from $\{red, blue, green\}$ such that, for each edge, its endpoints are of different colors.

  Note that even though $\check{G}$ is an undirected graph, we will represent each edge as a tuple $(u, v) \in \check{E}$ where $u, v \in \check{V}$. We will never swap the order of the two endpoints within this proof.

  We create a projection game $(A, B, E, \Sigma_A, \Sigma_B, \{\pi_e\}_{e \in E})$ as follows:
  \begin{itemize}
  \item Let $A = \check{E}$ and $B = \check{V}$.
  \item $E = \{(a, b) \in A \times B\mid b \text{ is an endpoint of } a \text{ with respect to } \check{G}\}$.
  \item $\Sigma_A = \{(red, blue), (red, green), (blue, red), (blue, green), (green, red), (green, blue)\}$ and $\Sigma_B = \{red, blue, green\}$.
  \item For each $e = (u, v) \in \check{E} = A$, let $\pi_{(e, u)}: (c_1, c_2) \to c_1$ and $\pi_{(e, v)}: (c_1, c_2) \to c_2$, i.e., $\pi_{(e, u)}$ and $\pi_{(e, v)}$ are projections to the first and the second element of the tuple respectively.
  \end{itemize}

  It is obvious that $G = (A, B, E)$ is a planar graph since $A, B$ are $\check{E}, \check{V}$ respectively and there exists an edge between $a \in A$ and $b \in B$ if and only if node corresponding to $b$ in $\check{V}$ is an endpoint of an edge corresponding to $a$ in $\check{E}$. This means that we can use the same planar embedding from the original graph $\check{G}$ except that each node represent a node from $B$ and at each edge, we put in a node from $A$ corresponding to that edge. It is also clear that the size of the projection game is polynomial of the size of $\check{G}$.

  The only thing left to show is to prove that $(A, B, E, \Sigma_A, \Sigma_B, \{\pi_e\}_{e \in E})$ is satisfiable if and only if $\check{G}$ is 3-colorable.

  ($\Rightarrow$) Suppose that $(A, B, E, \Sigma_A, \Sigma_B, \{\pi_e\}_{e \in E})$ is satisfiable. Let $\sigma_u$ be the assignment for each vertex $u \in A \cup B$ that satisfies all the edges in the projection game. We will show that by assigning $\sigma_v$ to $v$ for all $v \in \check{V} = B$, we are able to color $\check{G}$ with 3 colors such that, for each edge, its endpoints are of different color.

  Since $\check{V} = B$, $\sigma_v \in \{red, blue, green\}$ for all $v \in \check{V}$. Thus, this is a valid coloring. To see that no two endpoints of any edge are of the same color, consider an edge $e = (u, v) \in \check{E} = A$. From definition of $E$, we have $(e, u) \in E$ and $(e, v) \in E$. Moreover, from definition of $\pi_{(e, u)}$ and $\pi_{(e, v)}$, we can conclude that $\sigma_e = (\sigma_u, \sigma_v)$. Since $\sigma_e \in \Sigma_A$, we can conclude that $\sigma_u \ne \sigma_v$ as desired.

  Thus, $\check{G}$ is 3-colorable.

  ($\Leftarrow$) Suppose that $\check{G}$ is 3-colorable. In a valid coloring scheme, let $c_v$ be a color of node $v$ for each $v \in \check{V} = B$. Define the assignment of the projection game $\varphi_A, \varphi_B$ as follows
  \begin{align*}
    \varphi_A(a) = (c_u, c_v) &\text{ for all } a = (u, v) \in A = \check{E},\\
    \varphi_B(b) = c_b &\text{ for all } b \in B = \check{V}.
  \end{align*}

  Since $c_u \ne c_v$ for all $(u, v) \in \check{E}$, we can conclude that the range of $\varphi_A$ is a subset of $\Sigma_A$. Moreover, it is clear that the range of $\varphi_B$ is a subset of $\Sigma_B$. As a result, the assignment defined above is valid. Moreover, it is obvious that $\pi_e(\varphi_A(a)) = \varphi_B(b)$ for all $e = (a, b) \in E$. Hence, the projection game $(A, B, E, \Sigma_A, \Sigma_B, \{\pi_e\}_{e \in E})$ is satisfiable.

  As a result, we can conclude that \textsc{Label Cover} on planar graph is NP-hard.

\end{proof}

\subsection{PTAS for Projection Games on Planar Graphs} \label{subsec:ptas-planar}

We use the standard Baker's technique for finding PTAS for problems on planar graphs~\cite{Baker94} to construct one for projection games. 
Although not realized in the original paper by Baker, the technique relies on the concept of {\em treewidth}, which we will review the definition next. In this perspective, Baker's technique constructs an algorithm based on two main steps:
\begin{enumerate}
\item {\bf Thinning Step:} given a graph $G = (V, E)$, partition $E$ into subsets $S_1, \dots, S_h$ such that, for each $i = 1, \dots, h$, we obtain a graph with bounded treewidth when all edges in $S_i$ are deleted from the original graph.
\item {\bf Dynamic Programming Step:} for each $i = 1, \dots, h$, use dynamic programming to solve the problem on $(V, E - S_i)$, which has bounded treewidth. Then, output the best solution among these $h$ solutions.
\end{enumerate}

\subsubsection{Tree Decomposition and Treewidth}

Before proceed to the algorithm, we first define \emph{tree decomposition}. A \emph{tree decomposition} of a graph $G = (V, E)$ is a collection of subsets $B_1, \dots, B_n \subseteq V$ and a tree $T$ on these subsets such that
\begin{enumerate}
\item $V = B_1 \cup \cdots \cup B_n.$
  \item For each edge $(u, v) \in E$, there exists $B_i$ such that $u, v \in B_i$.
  \item For each $B_i$ and $B_j$, if $v \in B_i \cap B_j$, then $v$ must be in every subset along the path in $T$ from $B_i$ to $B_j$.
\end{enumerate}

The \emph{width} of a tree decomposition $(\{B_1, \dots, B_n\}, T)$ is the largest size of $B_1, \dots, B_n$ minus one. The \emph{treewidth} of a graph is the minimum width across all possible tree decompositions.

\subsubsection{Thinning}

Even though a planar graph does not necessarily have a bounded treewidth, it is possible to delete a ``small'' set of edges from the graph to obtain a graph with bounded treewidth; by ``small'', we do not refer to the size but refer to the change in the optimal solution for the projection game after we delete the edges from the set.

To achieve this, we partition $E$ into $h$ sets such that, when deleting all edges from each set, the treewidth is bounded linearly on $h$. Later on, we will show that, for at least one of the sets, deleting all the edges from it affects the optimal solution of projection games by at most a factor of $1 - 1/h$.

Baker implicitly proved the following partitioning lemma in her paper~\cite{Baker94}. For a more explicit formulation, please see~\cite{Epp00}. \footnote{In both~\cite{Baker94} and~\cite{Epp00}, the vertices, not edges, are partitioned. However, it is obvious that the lemma works for edges as well.}
\begin{lemma}[\cite{Baker94}]
  For any planar graph $G = (V, E)$ and integer $h$, there is a $O(h(|V| + |E|))$-time algorithm that outputs a partition $(S_1, \dots, S_h)$ of $E$ ,and, for each $i$, a tree decomposition of $(V, E - S_i)$ having width at most $O(h)$.
\end{lemma}


Next, we will show that, for at least one of the $S_i$'s, removing all the edges from it affects the optimal solution by at most a factor of $1 - 1/h$:

\begin{lemma}
  For any projection games instance on graph $G = (V, E)$ and any partition $(S_1, \dots, S_h)$ of $E$, there exists $i \in \{1, \dots, h\}$ such that the projection game instance resulted from removing all the edges in $S_i$ has the optimal solution that is within $1 - 1/h$ factor of the optimal of the original instance.
\end{lemma}

\begin{proof}
  Suppose that $E_{sat}$ is the set of all the edges satisfied by the optimal assignment $\varphi_{OPT}$ of the original instance. From the pigeonhole principle, there exists at least one $i \in \{1, \dots, h\}$ such that $|S_i \cap E_{sat}| \leq |E_{sat}|/h$. Since $\varphi_{OPT}$ still satisfies all the edges in $E - (S_i \cap E)$ in the projection game instance induced by $(V, E - S_i)$, we can conclude that the optimal assignment to this instance satisfies at least $(1 - 1/h)|E_{sat}|$ edges, which concludes the proof of this lemma.
\end{proof}

For the purpose of our algorithm, we select $h = 1 + \frac{1}{\varepsilon}$, which ensures that the treewidth of $(V, E - S_i)$ is at most $O(h) = O(1/\varepsilon)$ for each $i = 1, \dots, h$. Moreover, from the above lemma, we can conclude that, for at least one such $i$, the optimal solution of the projection game instance induced on $(V, E - S_i)$ satisfies at least $1 - 1/h = 1/(1 + \varepsilon)$ times as many edges satisfied by the optimal solution of the original instance.

\subsubsection{Dynamic Programming}

Next, we present a dynamic programming algorithm that solves a projection game in a bounded treewidth bipartite graph $G' = (A', B', E')$, given its tree decomposition $(\{B_1, \dots, B_n\}, T)$ of width $w$.


At a high level, the algorithm works as follows. We use depth-first search starting at $B_1$ to traverse the tree $T$. At each node $B_i$ and each assignment $\varphi : B_i \to (\Sigma_A \cup \Sigma_B)$ of $B_i$, the subproblem is to find the maximum number of satisfied edges with both endpoints in one of the nodes in the subtree of $T$ rooted at $B_i$ when all the vertices in $B_i$ are assigned according to $\varphi$. Let us call the answer to this subproblem $a_{B_i, \varphi}$. After finish solving all such subproblems, we go through all assignments $\varphi$'s for $B_1$ and output the maximum $a_{B_1, \varphi}$ among such $\varphi$'s. This is the optimum of the projection game.

Now, we give details on how to find each $a_{B_i, \varphi}$ based on the previously solved subproblems. We start by setting $a_{B_i, \varphi}$ to be the number of edges with both endpoints in $B_i$ that are satisfied by $\varphi$. Then, for each child $B_j$ of $B_i$, find the maximum value of $a_{B_j, \varphi'}$ among all ``compatible'' assignments $\varphi'$ of $B_j$, i.e., $\varphi'$ such that $\varphi(u) = \varphi'(u)$ for all $u \in B_i \cap B_j$. For each $B_j$, add this value to $a_{B_i, \varphi}$. Lastly, for each $B_j$, find the number of edges with both endpoints in $B_i \cap B_j$ that are satisfied by $\varphi$ and subtract this from $a_{B_i, \varphi}$.

The correctness of the algorithm follows easily from the description of the algorithm. To analyze the running time of the algorithm, first observe that there are $(|\Sigma_A|+|\Sigma_B|)^{|B_i|} \leq (|\Sigma_A|+|\Sigma_B|)^{w+1}$ possible assignments for $B_i$.Thus, there are at most $n(|\Sigma_A|+|\Sigma_B|)^{w+1}$ subproblems.

Moreover, observe that the running time of the algorithm is dominated by the time spent to compute the maximum value of $a_{B_j, \varphi'}$ for each child $B_j$. Since $B_i$ has at most $n$ children and there are at most $(|\Sigma_A|+|\Sigma_B|)^{w+1}$ assignments for $B_j$, there are at most $n(|\Sigma_A|+|\Sigma_B|)^{w+1}$ such $(B_i, \varphi)$ pairs to enumerate for each $(B_i, \varphi)$. Hence, the running time to solve each subproblem is $O(n(|\Sigma_A|+|\Sigma_B|)^{w+1})$.

As a result, the overall running time for the dynamic programming algorithm is $O(n^2(|\Sigma_A|+|\Sigma_B|)^{2w+2})$.


\subsubsection{Summary}

We use the dynamic programming algorithm presented above to solve the projection game instance induced by the graph $G_i = (V, E - S_i)$ for each $i = 1, \dots, h$. We then output the solution that satisfies most edges among all $i$'s. As shown earlier, since we select $h$ to be $1 + \frac{1}{\varepsilon}$, at least one of the solution found on $G_i$'s satisfies at least $\frac{1}{1 + \varepsilon}$ times as many edges as the optimal solution, which means that our algorithm is indeed an $(1 + \varepsilon)$-approximation algorithm.


Finally, since we use the dynamic programming algorithm $h$ times, the running time of the algorithm is $O(h n^2(|\Sigma_A|+|\Sigma_B|)^{2w+2}) = (nk)^{O(w + h)} = (nk)^{O(h)} = (nk)^{O(1/\varepsilon)}$. This gives us the desired PTAS for projection games on planar graphs.

\subsection{PTAS Running Time Lower Bound for Projection Games on Planar Graphs} \label{s:planar}

We devote this subsection to prove the running time lower bound for PTAS for projection games on planar graphs as formalized in a theorem below.
\begin{theorem} \label{thm-PTAS}
If ETH holds, then there is no PTAS for projection games on planar graphs running in time $2^{O(1/\varepsilon)^\gamma}(nk)^{O(1/\varepsilon)^{1-\delta}}$ for any constants $\gamma, \delta > 0$.
\end{theorem}
The theorem essentially means that the PTAS in the previous subsection cannot be substantially improved in terms of running time unless ETH is false.

The main idea of our proof is a reduction from {\sc Matrix Tiling} problem introduced by Marx who has successfully used the problem as a basis for proving PTAS running time lower bounds for many problems~\cite{Mar07}. Before we proceed to the proof, we start by reviewing the definition of {\sc Matrix Tiling} and stating a theorem from~\cite{Mar07} that we will use in the proof.

The {\sc Matrix Tiling} problem can be defined as follows.

\textsc{Input:} Positive integers $\tilde{k}, \tilde{n}$ and sets $S_{i, j} \subseteq [\tilde{n}] \times [\tilde{n}]$ for each $i, j = 1, \dots, \tilde{k}$.

\textsc{Goal:} Select $s_{i, j} \in S_{i, j} \cup \{\bigstar\}$ for every $i, j = 1 \dots, \tilde{k}$ such that
\begin{itemize}
  \item for every $i \in [\tilde{k}], j \in [\tilde{k} - 1]$, if $s_{i, j}, s_{i, j + 1} \ne \bigstar$, then $(s_{i, j})_1 =  (s_{i, j + 1})_1$, and,
  \item for every $i \in [\tilde{k} - 1], j \in [\tilde{k}]$, if $s_{i, j}, s_{i + 1, j} \ne \bigstar$, then $(s_{i, j})_2 =  (s_{i + 1, j})_2$
 \end{itemize}
 that maximizes the number of $(i, j) \in [\tilde{k}] \times [\tilde{k}]$ such that $s_{i, j} \ne \bigstar$. \\

 Note here that $[\tilde{k}]$ denotes $\{1, \dots, \tilde{k}\}$ and $(s_{i, j})_1$ represents the value in the first coordinate of $s_{i, j}$. Similar notations in this section are defined in similar manners.

We extract a running time lower bound for approximating {\sc Matrix Tiling} from Theorem 2.3 in~\cite{Mar07} below.
\begin{lemma}[\cite{Mar07}] \label{lem-gridtilingoptlowerbound}
If ETH holds, then one cannot distinguish a {\sc Matrix Tiling} instance of optimum $\tilde{k}^2$ (i.e. none of $s_{i, j}$ is $\bigstar$) from that of optimum $\tilde{k}^2/(1 + \varepsilon')$ in time $2^{O(1/\varepsilon')^{\gamma}}(\tilde{k}\tilde{n})^{O(1/\varepsilon')^{1-\delta}}$ for any constants $\gamma, \delta > 0$. \footnote{In~\cite{Mar07}, the theorem is phrased as a running time lower bound for PTAS but it is clear from the proof that the theorem can be stated as our more specific version too.}
\end{lemma}

Now that we have stated the preliminaries, we are ready to describe the reduction from {\sc Matrix Tiling} to the projection games problem:
\begin{lemma} \label{lem-gridprojred}
	There is a polynomial-time reduction from a {\sc Matrix Tiling} instance to a {\sc Label Cover} instance with alphabet size $k =O(\tilde{n}^2)$ on planar graph with $n = O(\tilde{k}^2)$ vertices and $|E| = 4\tilde{k}^2 - 4\tilde{k}$ edges such that
	\begin{itemize}
	\item if the {\sc Matrix Tiling} instance is of optimum $\tilde{k}^2$, then the projection game instance is satisfiable, and,
	\item for any $0 \leq l \leq \tilde{k}^2 - 1$, if the {\sc Matrix Tiling} instance is of optimum $\tilde{k}^2 - l$, then the projection game is of optimum at least $|E| - l/2$ (i.e. at most $l/2$ edges are not satisfied in the optimal solution).
	\end{itemize}
\end{lemma}

\begin{proof}
The reduction proceeds as follows:
\begin{itemize}
  \item Let $A$ be a set containing $\tilde{k}^2$ vertices; call the vertices $a_{i, j}$ for all $i, j \in [\tilde{k}]$.
  \item Let $B$ be a set containing $2\tilde{k}^2 - 2\tilde{k}$ vertices; call the vertices $b_{i + 0.5, j}$ for all $i \in [\tilde{k} - 1], j = [\tilde{k}]$ and $b_{i, j + 0.5}$ for all $i \in [\tilde{k}], j = [\tilde{k} - 1]$.
  \item Let $E$ be $\{(a_{x, y}, b_{z, t}) \in A \times B \mid |x - z| + |y - t| = 0.5\}$.
  \item Let $\Sigma_A$ be $[\tilde{n}] \times [\tilde{n}]$.
  \item Let $\Sigma_B$ be $[\tilde{n}] \cup \{\blacksquare, \blacklozenge\}$.
  \item The projections $\pi_e$'s where $e = (a_{x, y}, b_{z, t})$ can be defined as follows.
  \begin{align*}
    \pi_{e}(s) =
    \begin{cases}
      s_1 &\text{if } s \in S_{x, y} \text{ and } x = z, \\
      s_2 &\text{if } s \in S_{x, y} \text{ and } y = t, \\
      \blacksquare &\text{if } s \notin S_{x, y}, x \geq z \text{ and } y \geq t,\\
      \blacklozenge &\text{if } s \notin S_{x, y}, x \leq z \text{ and } y \leq t,\\
    \end{cases}
  \end{align*}
  for all $(a_{x, y}, b_{z, t}) \in E$ and for all $s \in \Sigma_A$.
\end{itemize}

For an illustration of the reduction, please refer to Figure~\ref{fig:gridtiling-reduction} below.

\begin{figure}[h]
  \begin{center}
    \begin{tabular}{ c | c }
      \centering \includegraphics[scale=0.18]{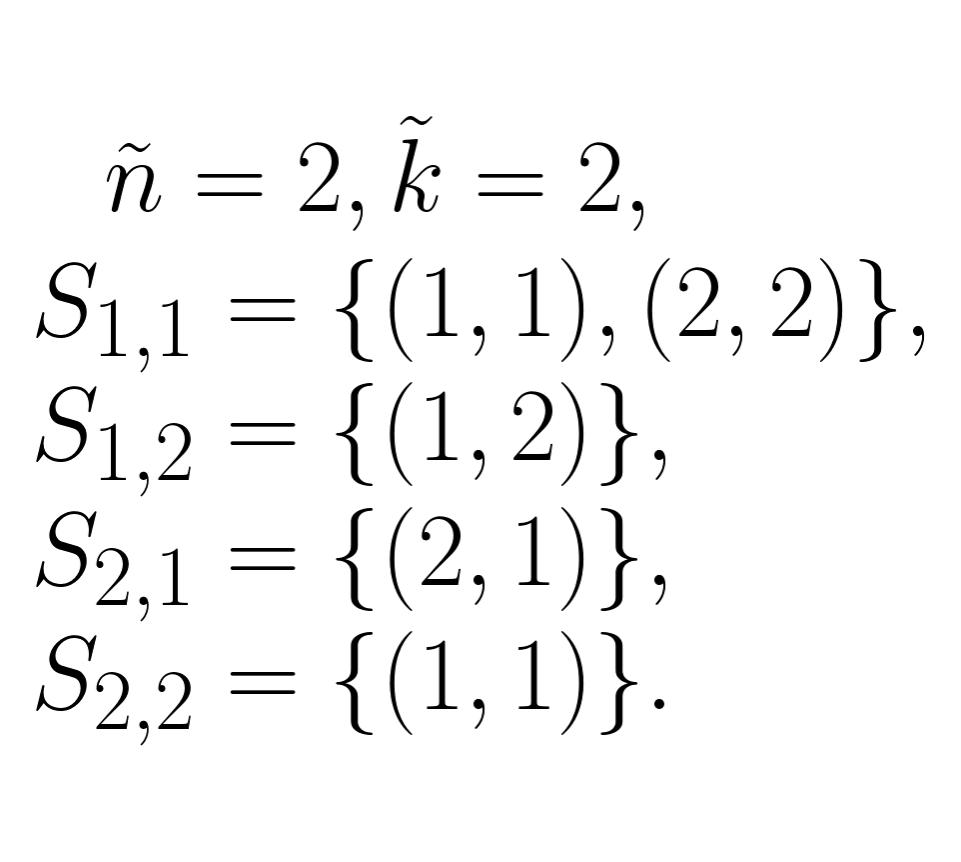} &  \includegraphics[scale=0.18]{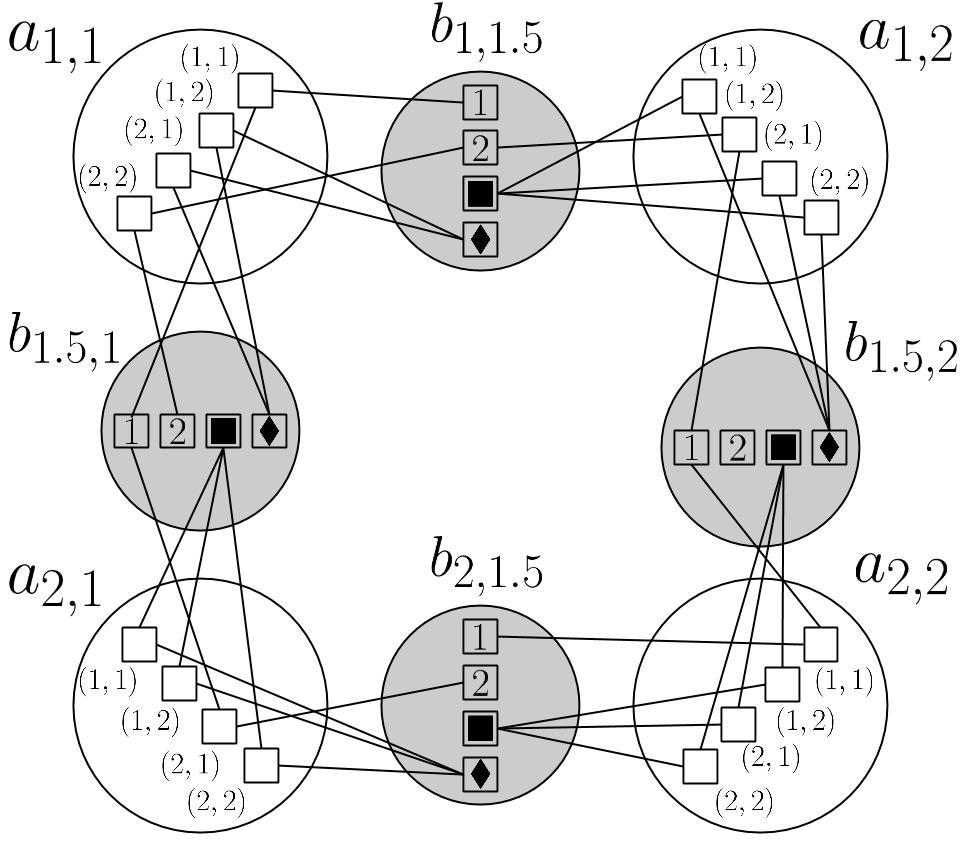} \\
      \\
    \end{tabular}
  \end{center}
  \caption{
    \textbf{An Example of the Reduction from {\sc Matrix Tiling} to {\sc Label Cover}.} The {\sc Matrix Tiling} instance is shown on the left and the projection game is on the right. For the projection game, circles represent vertices of the graph whereas squares represent alphabet symbols. Shaded circles are vertices in $B$ and white circles are those in $A$. Each line between two squares means that the corresponding projection maps the alphabet symbol from a vertex in $A$ to the alphabet symbol from a vertex in $B$.
  }
  \label{fig:gridtiling-reduction}
\end{figure}

It is obvious that the reduction runs in polynomial time, and that $k = O(\tilde{n}^2)$, $n = 3\tilde{k}^2 - 2\tilde{k} = O(\tilde{k}^2)$ and $|E| = 4\tilde{k}^2 - 4\tilde{k}$. Moreover, observe that $(A \cup B, E)$ is planar since, if we place $a_{x, y}$ on the plane at $(x, y)$ for all $a_{x, y} \in A$ and place $b_{z, t}$ at $(z, t)$ for all $b_{z, t} \in B$, then no edges intersect each other.

Now, we will prove the first property. Suppose that the {\sc Matrix Tiling} instance has optimum $\tilde{k}^2$. In other words, there exists $s_{i, j} \in S_{i, j}$ for all $i, j \in [\tilde{k}]$ such that, for each $i$, $(s_{i, j})_1$'s are equal for all $j  \in [\tilde{k}]$, and, for each $j$, $(s_{i, j})_2$'s are equal for all $i \in [\tilde{k}]$.

By simply picking $\varphi_A(a_{i, j}) = s_{i, j}$ for every $a_{i, j} \in A$, $\varphi_B(b_{i + 0.5, j}) = (s_{1, j})_2$ for every $b_{i + 0.5, j} \in B$ and, $\varphi_B(b_{i, j + 0.5}) = (s_{i, 1})_1$ for every $b_{i, j + 0.5} \in B$, we can conclude that the projection game is satisfiable.

Next, we will show the second property by contrapositive. Suppose that there is an assignment $\varphi_A: A \to \Sigma_A$ and $\varphi_A: B \to \Sigma_B$ that satisfies more than $|E| - l/2$ edges in the projection game. In other words, less than $l/2$ edges are not satisfied.

We create a solution to {\sc Matrix Tiling} instance as follows:
\begin{align*}
  s_{i, j} =
  \begin{cases}
    \varphi_A(a_{i, j}) &\text{if all edges with one endpoint in } N(a_{i, j}) \text{ are satisfied,} \\
    \bigstar &\text{otherwise.}
  \end{cases}
\end{align*}
for every $i, j \in [\tilde{k}]$. Note here that $N(a_{i, j})$ is a set of neighbors of $a_{i, j}$.

To see that this is a solution to the {\sc Matrix Tiling} instance, consider any $i, j$ such that $s_{i, j} \ne \bigstar$. If $i \leq \tilde{k} - 1$, from the definition of $s_{i, j}$, the edges $(a_{i, j}, b_{i + 0.5, j})$ and $(a_{i + 1, j}, b_{i + 0.5, j})$ are satisfied. This implies that $\varphi_A(a_{i, j}) \in S_{i, j}$ and $(\varphi_A(a_{i, j}))_2 = (\varphi_A(a_{i + 1, j}))_2$. In other words, $s_{i, j}$ and $s_{i + 1, j}$ will not contradict each other in {\sc Matrix Tiling}. Similarly, $s_{i, j}$ does not contradict with $s_{i - 1, j}, s_{i, j + 1}, s_{i, j - 1}$. Thus, the defined solution is a valid solution for {\sc Matrix Tiling}.

Next, since each unsatisfying edge can have an endpoint in $N(a_{i, j})$ for at most two pairs of $i, j \in [\tilde{k}]$, the number of $\bigstar$'s in a solution to {\sc Matrix Tiling} is at most $2$ times the number of unsatisfied edges in the projection game. Thus, the solution to {\sc Matrix Tiling} has less than $l$ $\bigstar$'s. In other words, the optimum of {\sc Matrix Tiling} instance is more than $\tilde{k}^2 - l$, which completes our proof.
\end{proof}

Finally, we will use Lemma~\ref{lem-gridtilingoptlowerbound} together with the reduction to prove Theorem~\ref{thm-PTAS}.
\begin{proof}[Proof of Theorem~\ref{thm-PTAS}]
  Suppose that there is a PTAS for planar projection games with running time $2^{O(1/\varepsilon)^\gamma}(nk)^{O(1/\varepsilon)^{1-\delta}}$ for some constants $\gamma, \delta > 0$.

  We will now use Lemma~\ref{lem-gridtilingoptlowerbound} to show that ETH fails. Given a {\sc Matrix Tiling} instance and $\varepsilon$. We use the reduction from Lemma~\ref{lem-gridprojred} to create a projection game instance with $k = O(\tilde{n}^2), n = O(\tilde{k}^2)$ and $|E| = 4\tilde{k}^2 - 4\tilde{k}$.

  From the first property of the reduction, if the {\sc Matrix Tiling} instance is of optimum $\tilde{k}^2$, then the projection game is satisfiable, i.e., is of optimum $4\tilde{k}^2 - 4\tilde{k}$. On the other hand, for any $\varepsilon' > 0$, from the second property, if the {\sc Matrix Tiling} instance is of optimum at most $\tilde{k}^2/(1 + \varepsilon') = \tilde{k}^2 - \varepsilon'\tilde{k}^2/(1 + \varepsilon')$, then the projection game is of optimum at most $4\tilde{k}^2 - 4\tilde{k} - \varepsilon'\tilde{k}^2/(2(1 + \varepsilon'))$.

  Pick $\varepsilon$ to be $\varepsilon'/20$. From our choice of $\varepsilon$, if $\varepsilon \leq 1$, then we have
  \begin{align*}
  	(4\tilde{k}^2 - 4\tilde{k})/(1 + \varepsilon) &= 4\tilde{k}^2 - 4\tilde{k} - \varepsilon (4\tilde{k}^2 - 4\tilde{k})/(1 + \varepsilon) \\
  	(\text{From our choice of } \varepsilon) &= 4\tilde{k}^2 - 4\tilde{k} - \varepsilon' (\tilde{k}^2 - \tilde{k})/(5(1 + \varepsilon'/20)) \\
  	&\geq 4\tilde{k}^2 - 4\tilde{k} - \varepsilon' \tilde{k}^2/(5(1 + \varepsilon'/20)) \\
  	(\varepsilon \leq 1) &= 4\tilde{k}^2 - 4\tilde{k} - \varepsilon' \tilde{k}^2/(2(1 + \varepsilon')).
  \end{align*}

  In other words, if we run the PTAS with the selected $\varepsilon$ on the projection game, then we are able to distinguish the game with optimum $4\tilde{k}^2 - 4\tilde{k}$ from that with optimum at most $4\tilde{k}^2 - 4\tilde{k} - \varepsilon'\tilde{k}^2/(2(1 + \varepsilon'))$. Hence, we can also distinguish a {\sc Matrix Tiling} instance of optimum $\tilde{k}^2$ from that of optimum $\tilde{k}^2/(1 + \varepsilon')$. Moreover, from our assumption, the running time of the PTAS is
  \begin{align*}
  	&2^{O(1/\varepsilon)^\gamma}(nk)^{O(1/\varepsilon)^{1-\delta}} \\
  	(\text{Since } k = O(\tilde{n}^2) \text{ and } n = O(\tilde{k}^2)) &= 2^{O(1/\varepsilon)^\gamma}(\tilde{n}\tilde{k})^{O(1/\varepsilon)^{1-\delta}} \\
  	(\text{Since } \varepsilon = \varepsilon'/20) &= 2^{O(1/\varepsilon')^\gamma}(\tilde{n}\tilde{k})^{O(1/\varepsilon')^{1-\delta}}.
  \end{align*}
  Thus, from Lemma~\ref{lem-gridtilingoptlowerbound}, ETH fails, which concludes the proof for Theorem~\ref{thm-PTAS}.
\end{proof}

\bibliographystyle{acm}

\bibliography{bi}

\newpage

{\huge \bf Appendix}

\appendix

\section{Polynomial-time Approximation Algorithms for Projection Games for Nonuniform Preimage Sizes}

In this section, we will describe a polynomial time $O((n_A|\Sigma_A|)^\frac{1}{4})$-approximation algorithm for satisfiable projection games, including those with nonuniform preimage sizes.

It is not hard to see that, if the $p_e$'s are not all equal, then ``know your neighbors' neighbors'' algorithm from Subsection~\ref{subsec-knowyourneighborsneighbors} does not necessarily end up with at least $h_{max}/\overline{p}$ fraction of satisfied edges anymore. The reason is that, for a vertex $a$ with large $|N_2(a)|$ and any assignment $\sigma_a \in \Sigma_A$ to the vertex, the number of preimages in $\pi_e^{-1}(\pi_{(a, b)}(\sigma_a))$ might be large for each neighbor $b$ of $a$ and each edge $e$ that has an endpoint $b$. We solve this issue, by instead of using all the edges for the algorithm, only using ``good'' edges whose preimage sizes for the optimal assignments are at most a particular value. However, this definition of ``good'' does not only depend on an edge but also on the assignment to the edge's endpoint in $B$, which means that we need to have some extra definitions to address the generalization of $h$ and $p$ as follows. \\

\begin{longtable}{l l}
$\sigma_b^{max}$ \hspace{8mm} & for each $b \in B$, denotes $\sigma_b \in \Sigma_B$ that maximizes the value of \\ & $\sum_{a \in N(b)} |\pi^{-1}_{(a, b)}(\sigma_b)|$. \\
$p^{max}_e$ & for each edge $e = (a, b)$, denotes $\left|\pi^{-1}_e(\sigma_b^{max})\right|$, the size of the \\ & preimage of $e$ if $b$ is assigned $\sigma_b^{max}$. \\
$\overline{p}^{max}$ & denotes the average of $p^{max}_e$ over all $e \in E$, i.e. $\frac{1}{|E|}\sum_{e \in E} p^{max}_e$. \\ & We will use $2\overline{p}^{max}$ as a threshold for determining ``good'' edges \\ & as we shall see below. \\
$E(S)$ & for each set of vertices $S$, denotes the set of edges with at least \\ & one endpoint in $S$, i.e. $\{(u, v) \in E \mid u \in S \text{ or } v \in S\}$. \\
$E_N^{max}$ & denotes the maximum number of edges coming out of $N(a)$ for \\ & all $a \in A$, i.e., $max_{a \in A}\{|E(N(a))|\}$. \\
$\Sigma^*_A(a)$ & for each $a \in A$, denotes the set of all assignments $\sigma_a$ to $a$ that, for \\ & every $b \in B$, there exists an assignment $\sigma_b$ such that, if $a$ is assigned $\sigma_a$, \\ & $b$ is assigned $\sigma_b$ and all $a$'s neighbors are assigned according to $a$, then \\ & there are still possible assignments left for all vertices in $N_2(a) \cap N(b)$, \\ & i.e., $\{\sigma_a \in \Sigma_A \mid \text{ for each } b \in B, \text{ there is } \sigma_b \in \Sigma_B  \text{ such that, for all } $ \\ & $ a' \in N_2(a) \cap N(b)\text{, } \left(\bigcap_{b' \in N(a') \cap N(a)} \pi^{-1}_{(a', b')}(\pi_{(a, b')}(\sigma_a))\right) \cap \pi^{-1}_{(a', b)}(\sigma_b) \ne \emptyset\}.$ \\ & Note that $\sigma^{OPT}_a \in \Sigma^*_A(a)$. In other words, if we replace $\Sigma_A$ with $\Sigma^*_A(a)$ \\ & for each $a \in A$, then the resulting instance is still satisfiable.  \\
$N^*(a, \sigma_a)$ & for each $a \in A$ and $\sigma_a \in \Sigma^*_A(a)$, denotes $\{b \in N(a) \mid |\pi^{-1}_{(a', b)}(\pi_{(a, b)}(\sigma_a))|$ \\ & $\leq 2\overline{p}^{max} \text{ for some } a' \in N(b)\}$. Provided that we assign $\sigma_a$ to $a$, this set \\ & contains all the neighbors of $a$ with at least one good edge as we \\ & discussed above. Note that $\pi_{(a, b)}(\sigma_a)$ is the assignment to $b$ \\ & corresponding to the assignment of $a$. \\
$N^*_2(a, \sigma_a)$ & for each $a \in A$ and $\sigma_a \in \Sigma^*_A(a)$, denotes all the neighbors of neighbors \\ & of $a$ with at least one good edge with another endpoint in $N(a)$ when $a$ \\ & is assigned $\sigma_a$, i.e., $\bigcup_{b \in N^*(a, \sigma_a)} \{a' \in N(b) \mid |\pi^{-1}_{(a', b)}(\pi_{(a, b)}(\sigma_a))| \leq 2\overline{p}^{max}\}$. \\
$h^*(a, \sigma_a)$ & for each $a \in A$ and $\sigma_a \in \Sigma^*_A(a)$, denotes $|E(N^*_2(a, \sigma_a))|$. In other words, \\ & $h^*(a, \sigma_a)$ represents how well $N^*_2(a, \sigma_a)$ spans the graph $G$. \\
$E^*(a, \sigma_a)$ & for each $a \in A$ and $\sigma_a \in \Sigma^*_A(a)$, denotes $\{(a', b) \in E \mid b \in N^*(a, \sigma_a),$ \\ & $a' \in N^*_2(a, \sigma_a) \text{ and } |\pi^{-1}_{(a', b)}(\pi_{(a, b)}(\sigma_a))| \leq 2\overline{p}^{max}\}$. When $a$ is assigned $\sigma_a$, \\ & this is the set of all good edges with one endpoint in $N(a)$. \\
$h^*_{max}$ & denotes $\max_{a \in A, \sigma_a \in \Sigma^*_A(a)} h^*(a, \sigma_a)$. \\
$E'$ & denotes the set of all edges $e \in E$ such that $p_e \leq 2\overline{p}^{max}$, i.e., \\ & $E' = \{e \in E \mid p_e \leq 2\overline{p}^{max}\}$. Recall that $p_e$ is defined earlier as $|\pi^{-1}(\sigma_b^{OPT})|$. \\ & Since $E'$ depends on $\sigma_b^{OPT}$, $E'$ will not be used in the algorithms below but \\ & only used in the analyses. Same goes for all the notations defined below. \\
$G'$ & denotes a subgraph of $G$ with its edges being $E'$. \\
$E'(S)$ & for each set of vertices $S$, denotes the set of all edges in $E'$ with \\ & at least one endpoint in $S$, i.e., $\{(u, v) \in E' \mid u \in S \text{ or } v \in S\}$. \\
$E'_S$ & for each set of vertices $S$, denotes the set of edges with both \\ & endpoints in $S$, i.e. $E'_S = \{(a, b) \in E' \mid a \in S \text{ and } b \in S\}$. \\
$N'(u)$ & for each vertex $u$, denotes the set of vertices that are neighbors of \\ & $u$ in the graph $G'$. \\
$N'(U)$ & for each set of vertices $U$, denotes the set of vertices that are \\ & neighbors of at least one vertex in $U$ in the graph $G'$. \\
$N'_2(u)$ & for each vertex $u$, denotes $N'(N'(u))$, the set of neighbors of \\ & neighbors of $u$ in $G'$. \\
\end{longtable}

From the definitions above, we can derive two very useful observations as stated below.

\begin{observation} \label{obs:edges}
  $|E'| \geq \frac{|E|}{2}$
\end{observation}

\begin{proof}
  Suppose for the sake of contradiction that $|E'| < \frac{|E|}{2}$. From the definition of $E'$, this means that, for more than $\frac{|E|}{2}$ edges $e$, we have $p_e > 2\overline{p}^{max}$. As a result, we can conclude that
  \begin{align*}
    |E|\overline{p}^{max} &< \sum_{e \in E} p_e \\
    &= \sum_{b \in B} \sum_{a \in N(b)} p_{(a, b)} \\
    &= \sum_{b \in B} \sum_{a \in N(b)} |\pi^{-1}_{(a, b)}(\sigma_b^{OPT})| \\
    &\leq \sum_{b \in B} \sum_{a \in N(b)} |\pi^{-1}_{(a, b)}(\sigma_b^{max})| \\
    &= |E|\overline{p}^{max}.
  \end{align*}
  This is a contradiction. Hence, $|E'| \geq \frac{|E|}{2}$.
\end{proof}

\begin{observation} \label{obs:optimalgprime}
  If $\sigma_a = \sigma_a^{OPT}$, then $N^*(a, \sigma_a) = N'(a)$, $N^*_2(a, \sigma_a) = N'_2(a)$ and $E^*(a, \sigma_a) = E'(N'(a))$.
\end{observation}

This observation is obvious since, when pluging in $\sigma_a^{OPT}$, each pair of definitions of $N^*(a, \sigma_a)$ and $N'(a)$, $N^*_2(a, \sigma_a)$ and $N'_2(a)$, and $E^*(a, \sigma_a)$ and $E'(N'(a))$ becomes the same.

Note also that from its definition, $G'$ is the graph with good edges when the optimal assignments are assigned to $B$. Unfortunately, we do not know the optimal assignments to $B$ and, thus, do not know how to find $G'$ in polynomial time. However, directly from the definitions above,  $\sigma_b^{max}, p_e^{max}, \overline{p}^{max}, E_N^{max}, \Sigma^*_A(a),$ $N^*(a, \sigma_a), N^*_2(a, \sigma_a)$, $h^*(a, \sigma_a)$ and $h^*_{max}$ can be computed in polynomial time. These notations will be used in the upcoming algorithms. Other defined notations we do not know how to compute in polynomial time and will only be used in the analyses.

For the nonuniform preimage sizes case, we use five algorithms as opposed to four algorithms used in uniform case. We will proceed to describe those five algorithms. In the end, by using the best of these five, we are able to produce a polynomial-time $O\left((n_A|\Sigma_A|)^{1/4}\right)$-approximation algorithm as desired.

We now list the algorithms along with their rough descriptions; detailed description and analysis of each algorithm will follow later on:
\begin{enumerate}
\item {\bf Satisfy one neighbor -- $|E|/n_B$-approximation.} Assign each vertex in $A$ an arbitrary assignment. Each vertex in $B$ is then assigned to satisfy one of its neighboring edges. This algorithm satisfies at least $n_B$ edges.

\item {\bf Greedy assignment -- ${|\Sigma_A|}/{\overline{p}^{max}}$-approximation.} Each vertex in $B$ is assigned an assignment $\sigma_b\in\Sigma_B$ that has the largest number of preimages across neighboring edges $\sum_{a \in N(b)} |\pi_{(a, b)}^{-1}(\sigma_b)|$. Each vertex in $A$ is then assigned so that it satisfies as many edges as possible. This algorithm works well when $\Sigma_B$ assignments have many preimages.

\item {\bf Know your neighbors -- $|E|/E_N^{max}$-approximation.} For a vertex $a_0 \in A$, pick an element of $\Sigma^*_A(a_0)$ and assign it to $a_0$. Assign its neighbors $N(a_0)$ accordingly. Then, for each node in $N_2(a_0)$, we find one assignment that satisfies all the edges between it and vertices in $N(a_0)$.

\item {\bf Know your neighbors' neighbors -- $O(|E|\overline{p}^{max}/h^*_{max})$-approximation.} For a vertex $a_0 \in A$, we go over all possible assignments in $\Sigma_A^*(a)$ to it. For each assignment, we assign its neighbors $N(a_0)$ accordingly. Then, for each node in $N_2(a_0)$, we keep only the assignments that satisfy all the edges between it and vertices in $N(a_0)$.

 When $a_0$ is assigned the optimal assignment, the number of choices for each node in $N^*_2(a_0)$ is reduced to at most $2\overline{p}^{max}$ possibilities. In this way, we can satisfy $1/2{\overline{p}^{max}}$ fraction of the edges that touch $N^*_2(a_0)$. This satisfies many edges when there exists $a_0\in A$ such that $N^*_2(a_0)$ spans many edges.

\item {\bf Divide and Conquer -- $O(n_A n_B (h^*_{max} + E_N^{max})/|E|^2)$-approximation.} For every $a\in A$, we can fully satisfy $N^*(a) \cup N^*_2(a)$ efficiently, and give up on satisfying other edges that touch this subset. Repeating this process, we can satisfy $\Omega(|E|^2/(n_A n_B (h^*_{max} + E_N^{max})))$ fraction of the edges. 
\end{enumerate}




Aside from the new ``know your neighbors'' algorithm, the main idea of each algorithm remains the same as in the uniform preimage sizes case. All the details of each algorithm are described below.

\subsubsection{Satisfy One Neighbor Algorithm.}

The algorithm is exactly the same as that of the uniform case.

\begin{lemma} \label{dBapprox-nonuniform}
  For satisfiable instances of projection games, an assignment that satisfies at least $n_B$ edges can be found in polynomial time, which gives the approximation ratio of $\frac{|E|}{n_B}$.
\end{lemma}

\begin{proof}
  The proof is exactly the same as that of Lemma~\ref{dBapprox}.
\end{proof}

\subsubsection{Greedy Assignment Algorithm.}

The algorithm is exactly the same as that of the uniform case.

\begin{lemma} \label{pickbest-nonuniform}
  There exists a polynomial-time $\frac{|\Sigma_A|}{\overline{p}^{max}}$-approximation algorithm for satisfiable instances of projection games.
\end{lemma}

\begin{proof}
  The proof of this lemma differs only slightly from the proof of Lemma~\ref{pickbest}.

  The algorithm works as follows:
  \begin{enumerate}
  \item For each $b$, assign it $\sigma^*_b$ that maximizes $\sum_{a \in N(b)} |\pi_{(a, b)}^{-1}(\sigma_b)|$.
  \item For each $a$, assign it $\sigma^*_a$ that maximizes the number of edges satisfied, $|\{b \in N(a) \mid \pi_{(a, b)}(\sigma_a) = \sigma^*_b\}|$.
  \end{enumerate}

  Let $e^*$ be the number of edges that get satisfied by this algorithm. We have
  \begin{align*}
    e^* &= \sum_{a \in A} |\{b \in N(a) \mid \pi_{(a, b)}(\sigma^*_a) = \sigma^*_b\}|.
  \end{align*}

  Due to the second step, for each $a \in A$, the number of edges satisfied is at least an average of the number of edges satisfied over all assignments in $\Sigma_A$. This can be written as follows.
  \begin{align*}
    e^* &= \sum_{a \in A} |\{b \in N(a) \mid \pi_{(a, b)}(\sigma^*_a) = \sigma^*_b\}| \\
    &\geq \sum_{a \in A} \frac{\sum_{\sigma_a \in \Sigma_A} |\{b \in N(a) \mid \pi_{(a, b)}(\sigma_a) = \sigma^*_b\}|}{|\Sigma_A|} \\
    &= \sum_{a \in A} \frac{\sum_{b \in N(a)} |\pi^{-1}_{(a, b)}(\sigma^*_b)|}{|\Sigma_A|} \\
    &= \frac{1}{|\Sigma_A|} \sum_{a \in A} \sum_{b \in N(a)} |\pi^{-1}_{(a, b)}(\sigma^*_b)|. \\
  \end{align*}

  From the definition of $\sigma_b^{max}$, we can conclude that $\sigma_b^* = \sigma_b^{max}$ for all $b \in B$. As a result, we can conclude that
  \begin{align*}
    e^* &\geq \frac{1}{|\Sigma_A|} \sum_{a \in A} \sum_{b \in N(a)} |\pi^{-1}_{(a, b)}(\sigma^*_b)| \\
    &= \frac{1}{|\Sigma_A|} \sum_{a \in A} \sum_{b \in N(a)} |\pi^{-1}_{(a, b)}(\sigma^{max}_b)| \\
    &= \frac{1}{|\Sigma_A|} \sum_{a \in A} \sum_{b \in N(a)} p^{max}_{(a, b)} \\
    &= \frac{1}{|\Sigma_A|} |E||\overline{p}^{max}| \\
    &= \frac{\overline{p}^{max}}{|\Sigma_A|} |E|.
  \end{align*}

  Hence, this algorithm satisfies at least $\frac{\overline{p}^{max}}{|\Sigma_A|}$ fraction of the edges, which concludes our proof.
\end{proof}

\subsection*{Know Your Neighbors Algorithm}
The next algorithm shows that one can satisfy all the edges with one endpoint in the neighbors of a vertex $a_0\in A$.

\begin{lemma} \label{knowyourneighbor-nonuniform}
  For each $a_0 \in A$, there exists a polynomial time $\frac{|E|}{|E(N(a_0))|}$-approximation algorithm for satisfiable instances of projection games.
\end{lemma}

\begin{proof}

  The algorithm works as follows:
  \begin{enumerate}
  \item Pick any assignment $\sigma_{a_0} \in \Sigma^*_A(a_0)$ and assign it to $a_0$:
  \item Assign $\sigma_b = \pi_{(a_0, b)}(\sigma_{a_0})$ to $b$ for all $b \in N(a_0)$.
  \item For each $a \in N_2(a_0)$, find the set of plausible assignments to $a$, i.e., $S_a = \{\sigma_a\in\Sigma_A \mid \forall b \in N(a) \cap N(a_0),  \pi_{(a, b)}(\sigma_a) =  \sigma_b\}$. Pick one $\sigma^*_a$ from this set and assign it to $a$. Note that $S_a \ne \emptyset$ from the definition of $\Sigma^*_A(a_0)$.
  \item Assign any assignment to unassigned vertices.
  \item Output the assignment $\{\sigma^*_{a}\}_{a\in A}$, $\{\sigma^*_b\}_{b\in B}$ from the previous step.
    \end{enumerate}

  From step 3, we can conclude that all the edges in $E(N(a_0))$ get statisfied. This yields $\frac{|E|}{|E(N(a_0))|}$ approximation ratio as desired.
\end{proof}

\subsection*{Know Your Neighbors' Neighbors Algorithm}

The next algorithm shows that if the neighbors of neighbors of a vertex $a_0\in A$ expand, then one can satisfy many of the (many!) edges that touch the neighbors of $a_0$'s neighbors. While the core idea is similar to the uniform version, in this version, we will need to consider $N^*_2(a_0, \sigma_{a_0})$ instead of $N_2(a_0)$ in order to ensure that the number of possible choices left for each vertex in this set is at most $2\overline{p}^{max}$.

\begin{lemma} \label{reduction-nonuniform}
  For each $a_0 \in A$ and $\sigma_{a_0} \in \Sigma_A^*(a_0)$, there exists a polynomial-time $O\left(\frac{|E|\overline{p}^{max}}{h^*(a_0, \sigma_{a_0})}\right)$-approximation algorithm for satisfiable instances of projection games.
\end{lemma}

\begin{proof}
  To prove Lemma~\ref{reduction-nonuniform}, we first fix $a_0 \in A$ and $\sigma_{a_0} \in \Sigma_A^*(a_0)$. We will describe an algorithm that satisfies $\Omega\left(\frac{h^*(a_0, \sigma_{a_0})}{\overline{p}^{max}}\right)$ edges, which implies the lemma.

  The algorithm works as follows:
  \begin{enumerate}
  \item Assign $\sigma_b = \pi_{(a_0, b)}(\sigma_{a_0})$ to $b$ for all $b \in N(a_0)$.
  \item For each $a \in A$, find the set of plausible assignments to $a$, i.e., $S_a = \{\sigma_a \in \Sigma_A \mid \forall b \in N(a) \cap N(a_0),  \pi_{(a, b)}(\sigma_a) =  \sigma_b\}$. Note that $S_a \ne \emptyset$ from the definition of $\Sigma^*_A(a_0)$.
  \item For all $b \in B$, pick an assignment $\sigma^*_b$ for $b$ that maximizes the average number of satisfied edges over all assignments in $S_a$ to vertices $a$ in $N(b) \cap N^*_2(a_0)$, i.e., maximizes $\sum_{a \in N(b) \cap N^*_2(a_0)} |\pi^{-1}_{(a, b)}(\sigma_b) \cap S_a|$.
  \item For each vertex $a \in A$, pick an assignment $\sigma^*_a \in S_a$ that maximizes the number of satisfied edges, $|\{b \in N(a) \mid \pi_{(a, b)}(\sigma_a) = \sigma_b^*\}|$ over all $\sigma_a \in S_a$.
  \end{enumerate}

  We will prove that this algorithm indeed satisfies at least $\frac{h^*(a_0, \sigma_{a_0})}{\overline{p}^{max}}$ edges.

  Let $e^*$ be the number of edges satisfied by the algorithm.  We have
  \begin{align*}
    e^* &= \sum_{a \in A} |\{b \in N(a) \mid \pi_{(a, b)}(\sigma^*_a) = \sigma^*_b\}|.
  \end{align*}

  Since for each $a \in A$, the assignment $\sigma^*_a$ is chosen to maximize the number of edges satisfied, we can conclude that the number of edges satisfied by selecting $\sigma^*_a$ is at least the average of the number of edges satisfied over all $\sigma_a \in S_a$.

  As a result, we can conclude that
  \begin{align*}
    e^* &\geq \sum_{a \in A} \frac{\sum_{\sigma_a \in S_a} |\{b \in N(a) \mid \pi_{(a, b)}(\sigma_a) = \sigma^*_b\}|}{|S_a|} \\
    &= \sum_{a \in A} \frac{\sum_{\sigma_a \in S_a} \sum_{b \in N(a)} 1_{\pi_{(a, b)}(\sigma_a) = \sigma^*_b}}{|S_a|} \\
    &= \sum_{a \in A} \frac{\sum_{b \in N(a)} \sum_{\sigma_a \in S_a} 1_{\pi_{(a, b)}(\sigma_a) = \sigma^*_b}}{|S_a|} \\
    &= \sum_{a \in A} \frac{\sum_{b \in N(a)} |\pi_{(a, b)}^{-1}(\sigma^*_b) \cap S_a|}{|S_a|} \\
    &= \sum_{b \in B} \sum_{a \in N(b)} \frac{|\pi_{(a, b)}^{-1}(\sigma^*_b) \cap S_a|}{|S_a|} \\
    &\geq \sum_{b \in B} \sum_{a \in N(b) \cap N^*_2(a_0, \sigma_{a_0})} \frac{ |\pi_{(a, b)}^{-1}(\sigma^*_b) \cap S_a|}{|S_a|}
  \end{align*}

  From the definition of $N^*_2(a_0, \sigma_{a_0})$, we can conclude that, for each $a \in N^*_2(a_0, \sigma_{a_0})$, there exists $b' \in N^*(a_0) \cap N(a)$ such that $|\pi^{-1}_{(a, b')}(\sigma_{b'})| \leq 2\overline{p}^{max}$. Moreover, from the definition of $S_a$, we have $S_a \subseteq \pi^{-1}_{(a, b')} (\sigma_{b'})$. As a result, we can arrive at the following inequalities.
  \begin{align*}
    |S_a| &\leq |\pi^{-1}_{(a, b')} (\sigma_{b'})| \\
    &\leq 2\overline{p}^{max}.
  \end{align*}

  This implies that
  \begin{align*}
    e^* &\geq \frac{1}{2\overline{p}^{max}} \sum_{b \in B} \sum_{a \in N(b) \cap N^*_2(a_0, \sigma_{a_0})} |\pi_{(a, b)}^{-1}(\sigma^*_b) \cap S_a|.
  \end{align*}

  From the definition of $\Sigma_A^*(a_0)$, we can conclude that, for each $b \in B$, there exists $\sigma_b \in B$ such that $\pi_{(a, b)}^{-1}(\sigma_b) \cap S_a \ne \emptyset$ for all $a \in N_2(a_0) \cap N(b)$. Since $N^*_2(a_0, \sigma_{a_0}) \subseteq N_2(a_0)$, we can conclude that $|\pi_{(a, b)}^{-1}(\sigma_b) \cap S_a| \geq 1$ for all $a \in N^*_2(a_0, \sigma_{a_0}) \cap N(b)$.

  Since we pick the assignment $\sigma^*_b$ that maximizes $\sum_{a \in N(b) \cap N^*_2(a_0)} |\pi^{-1}_{(a, b)}(\sigma^*_b) \cap S_a|$ for each $b \in B$, we can conclude that
  \begin{align*}
    e^* &\geq \frac{1}{2\overline{p}^{max}} \sum_{b \in B} \sum_{a \in N(b) \cap N^*_2(a_0, \sigma_{a_0})} |\pi_{(a, b)}^{-1}(\sigma^*_b) \cap S_a|\\
    &\geq \frac{1}{2\overline{p}^{max}} \sum_{b \in B} \sum_{a \in N(b) \cap N^*_2(a_0, \sigma_{a_0})} |\pi_{(a, b)}^{-1}(\sigma_b) \cap S_a| \\
    &\geq \frac{1}{2\overline{p}^{max}} \sum_{b \in B} \sum_{a \in N(b) \cap N^*_2(a_0, \sigma_{a_0})} 1.
  \end{align*}

  The last term can be rewritten as
  \begin{align*}
    \frac{1}{2\overline{p}^{max}} \sum_{b \in B} \sum_{a \in N(b) \cap N^*_2(a_0, \sigma_{a_0})} 1 &= \frac{1}{2\overline{p}^{max}} \sum_{a \in N^*_2(a_0, \sigma_{a_0})} \sum_{b \in N(a)} 1 \\
    &= \frac{1}{2\overline{p}^{max}} \sum_{a \in N^*_2(a_0, \sigma_{a_0})} d_a \\
    &= \frac{h^*(a_0, \sigma_{a_0})}{2\overline{p}^{max}}.
  \end{align*}

  As a result, we can conclude that this algorithm gives an assignment that satisfies at least $\frac{h^*(a_0, \sigma_{a_0})}{2\overline{p}^{max}}$ edges out of all the $|E|$ edges. Hence, this is a polynomial-time $O\left(\frac{|E|\overline{p}^{max}}{h^*(a_0, \sigma_{a_0})}\right)$-approximation algorithm as desired.
\end{proof}

\subsubsection{Divide and Conquer Algorithm.}

We will present an algorithm that separates the graph into disjoint subgraphs for which we can find the optimal assignments in polynomial time. We shall show below that, if $h^*(a, \sigma_a)$ is small for all $a \in A$ and $\sigma_a \in \Sigma^*_A(a)$, then we are able to find such subgraphs that contain most of the graph's edges.

\begin{lemma} \label{separation-nonuniform} There exists a polynomial-time $O\left(\frac{n_An_B(h^*_{max} + E_N^{max})}{|E|^2}\right)$-approximation algorithm for satisfiable instances of projection games.
\end{lemma}

\begin{proof}
  To prove this lemma, we will present an algorithm that gives an assignment that satisfies $\Omega\left(\frac{|E|^3}{n_An_B(h^*_{max} + E_N^{max})}\right)$ edges.

  We use $\mathcal{P}$ to represent the collection of subgraphs we find. The family $\mathcal{P}$ consists of disjoint sets of vertices. Let $V_\mathcal{P}$ be $\bigcup_{P \in \mathcal{P}} P$.

  For any set $S$ of vertices, define $G_S$ to be the graph induced on $S$ with respect to $G$. Moreover, define $E_S$ to be the set of edges of $G_S$. We also define $E_\mathcal{P} = \bigcup_{P \in \mathcal{P}} E_P$. Note that $E_S$ is similar to $E'_S$ defined earlier in the appendix. The only difference is that $E'_S$ is with respect to $G'$ instead of $G$.

  The algorithm works as follows.
  \begin{enumerate}
  \item Set $\mathcal{P} \leftarrow \emptyset$.
  \item\label{a:while-nonuniform} While there exists a vertex $a \in A$ and $\sigma_a \in \Sigma^*_A(a)$ such that $$|E^*(a, \sigma_a) \cap E_{(A \cup B) - V_\mathcal{P}}| \geq \frac{1}{16} \frac{|E|^2}{n_An_B}:$$

    \begin{enumerate}
    \item\label{a:update-P-nonuniform} Set $\mathcal{P} \leftarrow \mathcal{P} \cup \{(N^*_2(a, \sigma_a) \cup N^*(a, \sigma_a)) - V_\mathcal{P}\}$.

    \end{enumerate}
  \item\label{a:assign-nonuniform} For each $P \in \mathcal{P}$, find in time $poly(|\Sigma_A|, |P|)$ an assignment to the vertices in $P$ that satisfies all the edges spanned by $P$. This can be done easily by assigning $\sigma_a$ to $a$ and $\pi_{(a, b)} (\sigma_a)$ to $b \in B \cap P$. Then assign any plausible assignment to all the other vertices in $A \cap P$.
  \end{enumerate}

  We will divide the proof into two parts. First, we will show that when we cannot find a vertex $a$ and an assignment $\sigma_a \in \Sigma^*_A(a)$ in step~\ref{a:while-nonuniform}, $\left|E_{(A \cup B) - V_\mathcal{P}} \right| \leq \frac{3|E|}{4}$. Second, we will show that the resulting assignment from this algorithm satisfies $\Omega\left(\frac{|E|^3}{n_An_B(h^*_{max} + E_N^{max})}\right)$ edges.

  We will start by showing that, if no vertex $a$ and an assignment $\sigma_a \in \Sigma^*_A(a)$ in step~\ref{a:while-nonuniform} exist, then $\left|E_{(A \cup B) - V_\mathcal{P}} \right| \leq \frac{3|E|}{4}$.

  Suppose that we cannot find a vertex $a$ and an assignment $\sigma_a \in \Sigma^*_A(a)$ in step~\ref{a:while-nonuniform}. In other words, $|E^*(a, \sigma_a) \cap E_{(A \cup B) - V_\mathcal{P}}| < \frac{1}{16} \frac{|E|^2}{n_An_B}$ for all $a \in A$ and $\sigma_a \in \Sigma^*_A(a)$.

  Since $\sigma^{OPT}_a \in \Sigma^*_A(a)$ for all $a \in A$, we can conclude that
  \begin{align*}
    |E^*(a, \sigma^{OPT}_a) \cap E_{(A \cup B) - V_\mathcal{P}}| < \frac{1}{16} \frac{|E|^2}{n_An_B}.
  \end{align*}

  From Observation~\ref{obs:optimalgprime}, we have $E^*(a, \sigma^{{OPT}}) = E'(N'(a))$. As a result, we have
  \begin{align*}
    \frac{1}{16} \frac{|E|^2}{n_An_B} &> |E^*(a, \sigma^{OPT}_a) \cap E_{(A \cup B) - V_\mathcal{P}}| \\
    &= |E'(N'(a)) \cap E_{(A \cup B) - V_\mathcal{P}}|
  \end{align*}
  for all $a \in A$.

  Since $E'(N'(a)) = E'_{N'(a) \cup N'_2(a)}$, we can rewrite the last term as
  \begin{align*}
    |E'(N'(a)) \cap E_{(A \cup B) - V_\mathcal{P}}| &= |E'_{N'(a) \cup N'_2(a)} \cap E_{(A \cup B) - V_\mathcal{P}}| \\
    &= |E'_{N'(a) \cup N'_2(a) - V_\mathcal{P}}|.
  \end{align*}

  Consider $\sum_{a \in A} |E'_{N'(a) \cup N'_2(a) - V_\mathcal{P}}|$. Since $|E'_{N'(a) \cup N'_2(a) - V_\mathcal{P}}| < \frac{1}{16} \frac{|E|^2}{n_An_B}$ for all $a \in A$, we have the following inequality:
  \begin{align*}
    \frac{|E|^2}{16n_B} > \sum_{a \in A} |E'_{N'(a) \cup N'_2(a) - V_\mathcal{P}}|.
  \end{align*}


  Let $N^p(v) = N'(v) - V_\mathcal{P}$ and $N_2^p(v) = N'_2(v) - V_\mathcal{P}$. Similary, define $N^p(S)$ for a subset $S \subseteq A \cup B$. It is easy to see that $N_2^p(v) \supseteq N^p(N^p(v))$. This implies that, for all $a \in A$, we have $|E'_{N^p(a) \cup N^p_2(a)}| \geq |E'_{N^p(a) \cup N^p(N^p(a))}|$. Moreover, it is easy to see that, for all $a \in A - V_\mathcal{P}$, we have $|E'_{N^p(a) \cup N^p(N^p(a))}| = \sum_{b \in N^p(a)} |N^p(b)|$.

  Thus, the following holds:
  \begin{align*}
    \sum_{a \in A} |E'_{(N'(a) \cup N'_2(a)) - V_\mathcal{P}}| &= \sum_{a \in A} |E_{(N^p(a) \cup N^p_2(a))}| \\
    &\geq \sum_{a \in A - V_\mathcal{P}} |E_{(N^p(a) \cup N^p_2(a))}| \\
    &= \sum_{a \in A - V_\mathcal{P}} \sum_{b \in N^p(a)} |N^p(b)| \\
    &= \sum_{b \in B - V_\mathcal{P}} \sum_{a \in N^p(b)} |N^p(b)| \\
    &= \sum_{b \in B - V_\mathcal{P}} |N^p(b)|^2. \\
  \end{align*}

  From Jensen's inequality, we have
  \begin{align*}
    \sum_{a \in A} |E'_{(N'(a) \cup N'_2(a)) - V_\mathcal{P}}| &\geq \frac{1}{|B- V_\mathcal{P}|} \left(\sum_{b \in B - V_\mathcal{P}} |N^p(b)| \right)^2 \\
    &= \frac{1}{|B- V_\mathcal{P}|} \left|E'_{\left(A \cup B\right) - V_\mathcal{P}}\right|^2 \\
    &\geq \frac{1}{n_B} \left|E'_{\left(A \cup B\right) - V_\mathcal{P}}\right|^2.
  \end{align*}

  Since $\frac{|E|^2}{16n_B} \geq \sum_{a \in A} |E_{(N'(a) \cup N'_2(a)) - V_\mathcal{P}}|$ and $\sum_{a \in A} |E_{(N'(a) \cup N'_2(a)) - V_\mathcal{P}}| \geq  \frac{1}{n_B} \left|E'_{\left(A \cup B\right) - V_\mathcal{P}}\right|^2$, we can conclude that
  \begin{align*}
    \frac{|E|}{4} \geq \left|E'_{\left(A \cup B\right) - V_\mathcal{P}}\right|.
  \end{align*}

  Consider $E'_{\left(A \cup B\right) - V_\mathcal{P}}$ and $E_{\left(A \cup B\right) - V_\mathcal{P}}$. We have
  \begin{align*}
    E'_{\left(A \cup B\right) - V_\mathcal{P}} \cup (E - E') &\supseteq E_{\left(A \cup B\right) - V_\mathcal{P}} \\
    \left|E'_{\left(A \cup B\right) - V_\mathcal{P}}\right| + \left|E - E'\right| &\geq \left|E_{\left(A \cup B\right) - V_\mathcal{P}}\right| \\
    \frac{|E|}{4} + \left|E - E'\right| &\geq \left|E_{\left(A \cup B\right) - V_\mathcal{P}}\right|.
  \end{align*}

  From Observation \ref{obs:edges}, we have $|E'| \geq \frac{|E|}{2}$. Thus, we have
  \begin{align*}
    \frac{3|E|}{4} &\geq \left|E_{\left(A \cup B\right) - V_\mathcal{P}}\right|,
  \end{align*}
  which concludes the first part of the proof.

  Next, we will show that the assignment the algorithm finds satisfies at least $\Omega\left(\frac{|E|^3}{n_An_B(h^*_{max} + E_N^{max})}\right)$ edges. Since we showed that $\frac{3|E|}{4} \geq \left|E_{\left(A \cup B\right) - V_\mathcal{P}}\right|$ when the algorithm terminates, it is enough to prove that $|E_\mathcal{P}| \geq \frac{|E|^2}{16n_An_B(h^*_{max} + E_N^{max})} \left(|E| - \left|E_{\left(A \cup B\right) - V_\mathcal{P}}\right|\right)$. Note that the algorithm guarantees to satisfy all the edges in $E_\mathcal{P}$.

  We will prove this by using induction to show that at any point in the algorithm, $|E_\mathcal{P}| \geq \frac{|E|^2}{16n_An_B(h^*_{max} + E_N^{max})} \left(|E| - \left|E_{\left(A \cup B\right) - V_\mathcal{P}}\right|\right)$.

  \emph{Base Case.} At the beginning, we have $|E_\mathcal{P}| = 0 = \frac{|E|^2}{16n_An_B(h^*_{max} + E_N^{max})} \left(|E| - \left|E_{\left(A \cup B\right) - V_\mathcal{P}}\right|\right)$, which satisfies the inequality.

  \emph{Inductive Step.} The only step in the algorithm where any term in the inequality changes is step~\ref{a:update-P-nonuniform}. Let $\mathcal{P}_{old}$ and $\mathcal{P}_{new}$ be the set $\mathcal{P}$ before and after step~\ref{a:update-P-nonuniform} is executed, respectively. Let $a$ be the vertex selected in step~\ref{a:while-nonuniform}. Suppose that $\mathcal{P}_{old}$ satisfies the inequality.

  Since $|E_{\mathcal{P}_{new}}| = |E_{\mathcal{P}_{old}}| + |E_{(N^*(a, \sigma_a) \cup N^*_2(a, \sigma_a)) - V_{\mathcal{P}_{old}}}|$, we have
  \begin{align*}
    |E_{\mathcal{P}_{new}}| &= |E_{\mathcal{P}_{old}}| + |E_{(N^*(a, \sigma_a) \cup N^*_2(a, \sigma_a)) - V_{\mathcal{P}_{old}}}| \\
    &= |E_{\mathcal{P}_{old}}| + |E_{(N^*(a, \sigma_a) \cup N^*_2(a, \sigma_a))} \cap E_{(A \cup B) - V_{\mathcal{P}_{old}}}|.
  \end{align*}

  From the condition in step~\ref{a:while-nonuniform}, we have $|E^*(a, \sigma_a) \cap E_{(A \cup B) - V_{\mathcal{P}_{old}}}| \geq \frac{1}{16} \frac{|E|^2}{n_An_B}$. Moreover, $E_{(N^*(a, \sigma_a) \cup N^*_2(a, \sigma_a))} \supseteq E^*(a, \sigma_a)$ holds. As a result, we have
  \begin{align*}
    |E_{\mathcal{P}_{new}}| &= |E_{\mathcal{P}_{old}}| + |E_{(N^*(a, \sigma_a) \cup N^*_2(a, \sigma_a))} \cap E_{A \cup B - V_{\mathcal{P}_{old}}}| \\
    &\geq |E_{\mathcal{P}_{old}}| + |E^*(a, \sigma_a) \cap E_{(A \cup B) - V_{\mathcal{P}_{old}}}| \\
    &\geq |E_{\mathcal{P}_{old}}| + \frac{1}{16} \frac{|E|^2}{n_An_B}.
  \end{align*}

  Now, consider $\left(|E| - |E_{\left(A \cup B\right) - V_{\mathcal{P}_{new}}}|\right) - \left(|E| - |E_{\left(A \cup B\right) - V_{\mathcal{P}_{old}}}|\right)$. We have
  \begin{align*}
    \left(|E| - |E_{\left(A \cup B\right) - V_{\mathcal{P}_{new}}}|\right) - \left(|E| - |E_{\left(A \cup B\right) - V_{\mathcal{P}_{old}}}|\right) = |E_{\left(A \cup B\right) - V_{\mathcal{P}_{old}}}| - |E_{\left(A \cup B\right) - V_{\mathcal{P}_{new}}}|
  \end{align*}

  Since $V_{\mathcal{P}_{new}} = V_{\mathcal{P}_{old}} \cup \left(N^*_2(a, \sigma_a) \cup N^*(a, \sigma_a)\right)$, we can conclude that
  \begin{align*}
    \left((A \cup B) - V_{\mathcal{P}_{old}}\right) \subseteq \left((A \cup B) - V_{\mathcal{P}_{new}}\right) \cup \left(N^*_2(a, \sigma_a) \cup N^*(a, \sigma_a)\right).
  \end{align*}

  Thus, we can also derive
  \begin{align*}
    E_{(A \cup B) - V_{\mathcal{P}_{old}}} &\subseteq E_{\left((A \cup B) - V_{\mathcal{P}_{new}}\right) \cup \left(N^*_2(a, \sigma_a) \cup N^*(a, \sigma_a)\right)} \\
    &= E_{(A \cup B) - V_{\mathcal{P}_{new}}} \cup \{(a', b') \in E \mid a' \in N^*_2(a, \sigma_a) \text{ or } b' \in N^*(a, \sigma_a)\}. \\
  \end{align*}

  Moreover, we can write $\{(a', b') \in E \mid a' \in N^*_2(a, \sigma_a) \text{ or } b' \in N^*(a, \sigma_a)\}$ as $\{(a', b') \in E \mid a' \in N^*_2(a, \sigma_a)\} \cup \{(a', b') \in E \mid b' \in N^*(a, \sigma_a)\}$. Since $N^*(a, \sigma_a) \subseteq N(a)$, we can conclude that
  \begin{align*}
    \{(a', b') \in E \mid a' \in N^*_2(a, \sigma_a) \text{ or } b' \in N^*(a, \sigma_a)\} \subseteq & \{(a', b') \in E \mid a' \in N^*_2(a, \sigma_a)\} \\ & \cup \{(a', b') \in E \mid b' \in N(a)\}.
  \end{align*}

  Thus, we can conclude that
  \begin{align*}
    |\{(a', b') \in E \mid a' \in N^*_2(a, \sigma_a) \text{ or } b' \in N^*(a, \sigma_a)\}| &\leq |\{(a', b') \in E \mid a' \in N^*_2(a, \sigma_a)\}| \\ & \text{ } \text{ } \text{ } + |\{(a', b') \in E \mid b' \in N(a)\}| \\
    &= h^*(a, \sigma_a) + |E(N(a))|.
  \end{align*}

  Hence, we can conclude that
  \begin{align*}
    \left|E_{(A \cup B) - V_{\mathcal{P}_{old}}}\right| &\leq \left|E_{(A \cup B) - V_{\mathcal{P}_{new}}} \cup \{(a', b') \in E \mid a' \in N_2(a) \text{ or } b' \in N(a)\}\right| \\
    &\leq \left|E_{(A \cup B) - V_{\mathcal{P}_{new}}}\right| + \left|\{(a', b') \in E \mid a' \in N_2(a) \text{ or } b' \in N(a)\}\right| \\
    &\leq \left|E_{(A \cup B) - V_{\mathcal{P}_{new}}}\right| + h^*(a, \sigma_a) + |E(N(a))| \\
    &\leq \left|E_{(A \cup B) - V_{\mathcal{P}_{new}}}\right| + h^*_{max} + E_N^{max}.
  \end{align*}

  This implies that $\left(|E| - \left|E_{\left(A \cup B\right) - V_{\mathcal{P}}}\right|\right)$ increases by at most $h^*_{max} + E_N^{max}$.

  Hence, since $\left(|E| - \left|E_{\left(A \cup B\right) - V_{\mathcal{P}}}\right|\right)$ increases by at most $h^*_{max} + E_N^{max}$ and $\left|E_\mathcal{P}\right|$ increases by at least $\frac{1}{16}\frac{|E|^2}{n_An_B}$ and from the inductive hypothesis, we can conclude that
  \begin{align*}
    |E_{\mathcal{P}_{new}}| \geq \frac{|E|^2}{16n_An_B(h^*_{max} + E_N^{max})} \left(|E| - \left|E_{\left(A \cup B\right) - V_{\mathcal{P}_{new}}}\right|\right).
  \end{align*}

  Thus, the inductive step is true and the inequality holds at any point during the execution of the algorithm.

  When the algorithm terminates, since $|E_\mathcal{P}| \geq \frac{|E|^2}{16n_An_B(h^*_{max} + E_N^{max})} \left(|E| - \left|E_{\left(A \cup B\right) - V_\mathcal{P}}\right|\right)$ and $\frac{3|E|}{4} \geq \left|E_{\left(A \cup B\right) - V_\mathcal{P}}\right|$, we can conclude that $|E_{\mathcal{P}}| \geq \frac{|E|^3}{64n_An_B(h^*_{max} + E_N^{max})}$. Since the algorithm guarantees to satisfy every edge in $E_{\mathcal{P}}$, it yields an $O\left(\frac{n_An_B(h^*_{max} + E_N^{max})}{|E|^2}\right)$ approximation ratio, which concludes our proof of Lemma~\ref{separation-nonuniform}.
\end{proof}

\subsection*{Proof of Theorem~\ref{t:approx}}

\begin{proof}
  Using Lemma~\ref{reduction-nonuniform} with $a_0$ and $\sigma_{a_0}$ that maximizes the value of $h^*(a_0, \sigma_{a_0})$, i.e., $h^*(a_0, \sigma_{a_0}) = h^*_{max}$, we can conclude that there exists a polynomial-time $O\left(\frac{|E|\overline{p}^{max}}{h^*_{max}}\right)$-approximation algorithm for satisfiable instances of projection games.

  Similarly, from Leamma~\ref{knowyourneighbor-nonuniform} with $a_0$ that maximizes the value of $E(N(a_0))$, i.e., $|E(N(a_0))| = E_N^{max}$, there exists a polynomial-time $\frac{|E|}{E_N^{max}}$-approximation algorithm for satisfiable instances of projection games.

  Moreover, from Lemmas~\ref{dBapprox-nonuniform},~\ref{pickbest-nonuniform} and~\ref{separation-nonuniform}, there exists a polynomial-time $\frac{|E|}{n_B}$-approximation algorithm, a polynomial-time $\frac{|\Sigma_A|}{\overline{p}^{max}}$-approximation algorithm and a polynomial time $O\left(\frac{n_An_B(h^*_{max} + E_N^{max})}{|E|^2}\right)$-approximation algorithm for satisfiable instances of the projection game.

  Consider the following two cases.

  First, if $h^*_{max} \geq E_N^{max}$, we have $O(n_A n_B (h^*_{max} + E_N^{max})/|E|^2) = O(n_An_Bh^*_{max}/|E|^2)$. Using the best of the first, second, fourth and fifth algorithms, the smallest of the four approximation factors is at most as large as their geometric mean, i.e., $$O\left(\sqrt[4]{\frac{|E|}{n_B}\cdot\frac{|\Sigma_A|}{\overline{p}^{max}}\cdot \frac{|E|\overline{p}^{max}}{h^*_{max}}\cdot\frac{n_A n_B h^*_{max}}{|E|^2}}\right)= O((n_A|\Sigma_A|)^{1/4}).$$

  Second, if $E_N^{max} > h^*_{max}$, we have $O(n_A n_B (h^*_{max} + E_N^{max})/|E|^2) = O(n_An_BE_N^{max}/|E|^2)$. We use the best answer we get from the first, second, third and fifth algorithms. The smallest of the four approximation factors is at most as large as their geometric mean, i.e., $$O\left(\sqrt[4]{\frac{|E|}{n_B}\cdot\frac{|\Sigma_A|}{\overline{p}^{max}}\cdot \frac{|E|}{E_N^{max}} \cdot\frac{n_A n_B E_N^{max}}{|E|^2}}\right)= O\left(\left(\frac{n_A|\Sigma_A|}{\overline{p}^{max}}\right)^{1/4}\right).$$
  It is obvious that $\overline{p}^{max}$ is at least one. Thus, we can conclude that the approximation factor is at most $O((n_A|\Sigma_A|)^{\frac{1}{4}})$.

  This concludes the proof of Theorem~\ref{t:approx} for the nonuniform preimage sizes case.
\end{proof}

\end{document}